\documentclass{article}

\usepackage{amsmath,amssymb,latexsym,color,mathrsfs}
\usepackage{float}
\usepackage{soul}
\usepackage{appendix}
\usepackage{hyperref}
\hypersetup{
    colorlinks=true,
    linkcolor=blue,
    citecolor=blue,
    urlcolor=blue,
    pdfborder={0 0 0}
}

\newtheorem{theorem}{Theorem}[section]
\newtheorem{lemma}[theorem]{Lemma}

\newtheorem{claim}[theorem]{Claim}

\newtheorem{corollary}[theorem]{Corollary}

\newtheorem{definition}[theorem]{Definition}
\newtheorem{example}[theorem]{Example}

\newtheorem{proposition}[theorem]{Proposition}
\newtheorem{remark}{Remark}[section]

\newenvironment{proof}[1][Proof]{\noindent\textbf{#1.} }{\ \rule{0.5em}{0.5em}}

\begin{document}

\title{Continuity of equilibria in spaces of Bochner and Gel'fand economies}
\author{Mat\'ias Fuentes \\
  {\small Departamento de An\'alisis Econ\'omico: Econom\'ia Cuantitativa (UDI: Matem\'aticas)} \\
  {\small Universidad Aut\'onoma de Madrid} \\
  {\small Campus Universitario de Cantoblanco, 28049, Madrid, Spain} \\
  {\small \tt matias.fuentes@uam.es}
}
\date{}
\maketitle

\begin{abstract}
We examine the continuity of equilibrium correspondences in infinite-dimensional settings where the commodity spaces are Banach lattices. Economies are modeled as Borel probability measures on a space of characteristics, with aggregate endowments defined via Bochner or Gel'fand integrals. Within this framework, we prove that the equilibrium correspondence is continuous on a dense subset of the domain of economies admitting equilibria, endowed with a suitable Polish topology. These results extend both classical and recent continuity theorems by providing a unified analytical treatment applicable to a substantially broader class of locally convex spaces and encompass models with infinite planning horizons, monopolistic competition, neoclassical economies, financial equilibria, and asymmetric information. Importantly, this study demonstrates that there is no necessity to impose differentiability assumptions that are typically required in regular economies to study equilibrium continuity.

\ 


\ 


\

\noindent {\bf keywords}: Equilibrium correspondence; Banach lattices; Large economies; Distributional economies; Bochner integral; Gel'fand integral.

\end{abstract}

\section{Introduction}

Having established the existence of a competitive equilibrium, an examination of the equilibrium set could encompass inquiries into its uniqueness, regularity, or efficiency. In particular, the equilibrium set exhibits sensitivity to perturbations in exogenous parameters that define the agents and, consequently, the entire economic system. This interplay between parameters and equilibrium sets has been formalized in the literature through the concept of equilibrium correspondences, which map economies to their corresponding equilibria.

Building upon the foundational work of \cite{kannai_continuity_1970}, \cite{hildenbrand_economies_1970} and \cite{hildenbrand_upper_1972}, the present study delves into the continuity properties of the equilibrium correspondence for pure exchange economies. These works mentioned above, along with \cite{balasko_graph_1975}, all crucially recognize the role of exogenous parameters as defining characteristics of agents. This recognition is pivotal in establishing a topology within the space of economies.

The continuity of the equilibrium correspondence is a highly desirable property from both theoretical and applied perspectives. In its absence, seemingly minor measurement errors in the underlying parameters could lead to drastically different equilibrium sets. Consequently, the explanatory power of general equilibrium theory is critically dependent upon the continuity of the equilibrium correspondence.

The pioneering work of \cite{debreu_economies_1970} laid the groundwork for this analysis by proving that, in economies with regular properties, equilibria are locally unique if initial endowments are strictly positive. Furthermore, if the set of equilibria is compact, this is equivalent to saying that the equilibrium set is finite. The key mathematical tool used by Debreu was Sard's Theorem \cite{sard_measure_1942} under differentiability assumptions. We recall that a regular economy is an economy characterized by an excess demand function with the property that its slope at any equilibrium price vector is non-zero. This means that equilibrium prices are isolated and finite for any regular economy. Even more, \cite{debreu_economies_1970} proved that the set of economies that do not have a finite number of equilibria are negligible, i.e. of Lebesgue measure zero. Subsequently, \cite{hildenbrand_continuity_1972} strengthened Debreu's result by relaxing the assumption of strictly positive initial endowments.  

\cite{dierker_topological_1974} marked a significant progress by establishing the first result on the continuity of the equilibrium correspondence without differentiability assumptions (Corollary 8.5, p. 85).  The importance of relaxing that demand functions are continuously differentiable lies in its alignment with a key advancement in general equilibrium theory: the ability to move beyond the limitations of differential calculus. This achievement involved expanding the domain of analysis from regular economies to essential economies. In simpler terms, essential economies encompass all economies that possess essential equilibria. An essential equilibrium can be understood as one that can be approximated by equilibria from ``similar" economies. These similar economies are close to the original economy under a well-defined metric within the space of economies. The concept of essential equilibria, or essential stability, originates from fixed point theory through the work of Fort \cite{fort_essential_1950}, and its applicability has extended not only to general equilibrium theory but also to game theory (e.g., \cite{yang_hadamard_2017}). This broader applicability comes at a cost, however. Essential stability, unlike regularity, does not guarantee local isolation of equilibria. In simpler terms, essential equilibria may not be the only equilibria in a small neighborhood of a given economy.

Recent contributions to the study of the continuity of the equilibrium correspondence include the works of \cite{dubey_remark_2015} and \cite{he_modeling_2017}. The former investigates continuity by relaxing assumptions for economies with a finite number of agents, while the latter examines continuity within the framework of an atomless agent space structure. A significant constraint is the limitation to finite-dimensional commodity spaces, which restricts applicability to settings beyond the $\ell$-dimensional Euclidean space $\mathbb{R}^{\ell}$. Infinite-dimensional commodity spaces are essential for tackling problems in financial markets, dynamic general equilibrium theory, optimal growth theory, commodity differentiation, and the allocation of resources over time, encompassing both renewable and non-renewable ones. 

Building upon the previous works, \cite{cea_continuity_2024} recently extended the analysis to separable Banach spaces with a nonempty interior in the positive cone. Under mild assumptions, these authors proved that a topologically very large subset within the set of economies with equilibria exhibits continuous equilibria, that is, the equilibrium correspondence is continuous on that subset. This work provided an affirmative answer to a question raised by \cite{dubey_remark_2015} about the possibility of extending their result to infinite dimensional commodity space economies. The theoretical scope is rather limited due to the requirement of both a nonempty interior in the positive cone and separability, two properties satisfied by only a few topological spaces, namely, the space of continuous functions on a compact metric space $K$, $C(K)$, and weakly compact generated subspaces of $L_\infty$ and $\ell_\infty$, the spaces of real valued measurable essentially bounded functions on a finite measure space and the space of all bounded sequences respectively. From an economic point of view, it allows us to study the continuity of the equilibrium correspondence when considering forms of uncertainty, infinite horizon and commodity differentiation (see \cite{cea_continuity_2024}).

The main goal of the present work is to establish continuity results as in \cite{cea_continuity_2024} but encompassing many more economic models. Specifically, we focus on Banach lattices, which allows us to include most of the pertinent commodity spaces, offering a comprehensive approach to economic analysis as those of infinite plan horizon, monopolistic competition, neoclassical economies, financial equilibrium models and equilibria under asymmetric information, among others. Moreover, our study extends the scope to include not just subspaces of $L_\infty$ or $\ell_\infty$ as in \cite{cea_continuity_2024}, but the spaces themselves. This broader approach is made possible by avoiding the requirement of separability\footnote{It is noteworthy that while both $L_\infty$ and $\ell_\infty$ are not norm separable, their weakly compact generated subspaces are.}.

From a mathematical standpoint, our approach not only encompasses a wide range of topological vector spaces but also integrates diverse topologies within each space. This flexibility is critical because each economy is characterized as a Borel probability distribution over a space of characteristics, endowed with its corresponding Borel $\sigma$-algebra. The choice of the topology on this characteristic space determines the nature of its Borel 
$\sigma$-algebra, which is generated by (or associated with) the open sets of the characteristic space. Different topologies generally yield distinct Borel $\sigma$-algebras, thereby potentially giving rise to different economic models or scenarios depending on the underlying topology. This observation highlights the sensitivity of economic analyses to the mathematical structure chosen for the characteristic space.

In close association with this issue, our study contributes to advancing the treatment of total initial endowments in economic modeling. While \cite{cea_continuity_2024} made use of Bochner integrals over norm-compact spaces of initial endowments, our analysis extends to spaces where the Bochner integral is also defined over Mackey or weakly compact sets of initial endowments. This expansion enables a broader range of economic scenarios, accommodating more flexible variations in initial endowment structures. Furthermore, by incorporating dual spaces into our framework, we broaden even more our analytical scope, addressing diverse topological considerations in economic analysis. This extension allows us to consider economies where total endowments are defined through Gel'fand integrals over compact sets for different topologies.

We define the space of economies by a metric space that requires parameterizing the family of economies with respect to the dimensions of similarity. In our setting, these are the preference relations and the endowments termed the \emph{characteristics set}. Borel probability distributions on the characteristics set are a general way of defining economies at a macro level, which has been applied in many areas.  

To fully explore the ramifications of our continuity result, we transition to the analysis of individualized economies, that is, economies are defined as measurable mappings from an atomless measure space of agents directly to the characteristic set. This framework sheds light on the broader applicability of our continuity theorem. Crucially, it removes the common assumption in the literature that economies must share a homogenous agent space. Furthermore, this approach empowers us to establish two noteworthy properties: 1. any atomless probability space of consumers can be approached as the limit of a sequence with a progressively increasing number of finitely-typed agents. 2. The limit of a sequence of economies characterized by a finite and growing number of agent types can, itself, be an atomless probability space. 

This paper is structured as follows. Section \ref{maths} establishes the essential mathematical concepts employed throughout the analysis. Section \ref{sec:model} formally defines economic models, with a preliminary discussion on commodity-price pairings. Section \ref{economies} examines in detail the properties of distributional economies. Section \ref{continuity} introduces the equilibrium correspondence and presents our core results on continuity. Section \ref{individual} explores the alternative representation of individualized economies using measurable functions. Section \ref{applications} applies the theoretical framework developed in the preceding sections to analyze concrete economic problems. Finally, all formal proofs and special result are relegated to Appendices for improved readability.

\section{Mathematical preliminaries}\label{maths}

\ 

A subset of a topological space $(T, \mathcal T)$ is \emph{nowhere dense} if and only if its closure has empty interior. A subset of a topological space $(T, \mathcal T)$ is \emph{meager} in $T$ if and only if it is the countable union of nowhere dense sets. The complement of a meager set is called \emph{co-meager} or \emph{residual} in $T$. A subset which is not meager is called \emph{non-meager} in $T$. If every non-empty open subset is non-meager in $T$, then $T$ is called a \emph{Baire space}. 

\begin{theorem} \textnormal{(Baire - \cite{kelley_general_1955})} Let $T$ be either a complete pseudo-metric space or a locally compact regular space. Then the intersection of a countable family of open dense subsets of $T$ is itself dense in $T$.
\end{theorem}

The above theorem says that every locally compact space and every complete metrizable space is a Baire space. Further, it is straightforward to verify that the complement of a meager subset of complete metric space is dense. The next result is a variant of a classical result of \cite{fort_essential_1950} by \cite{wen_essential_1962}. This result will become very useful in our discussion of the continuity of equilibrium. 

\begin{lemma}\label{essential} Let $T$ and $M$ be a metric space and a Baire space respectively. Suppose further that $\Lambda : M \twoheadrightarrow T$ is a compact-valued and upper hemicontinuous correspondence with $\Lambda (m)\not = \emptyset$ for all $m\in M$. Then there exists a dense $G_{\delta}$ subset $M'$ of M such that $\Lambda$ is lower hemi-continuous at every point in $M'$.
\end{lemma}

\ 

Consider a Banach space $(L, \Vert\cdot\Vert )$. We shall call the norm topology simply the $\|\cdot\|$-topology. Let $L^*$ be the topological dual space of $L$, i.e. the space of continuous linear functionals on $L$. In addition to the norm topology, we consider two topologies consistent with the duality $\langle L, L^*\rangle$. The weakest one, the weak topology $\sigma (L, L^*)$ and the strongest one, the Mackey topology $\tau(L, L^*)$. We also shall consider on $L^*$ the norm, the weak$^*$ and the Mackey$^*$ topology denoted $\|\cdot\|^*$, $\sigma(L^*, L)$ and $\tau(L^*, L)$, respectively. 

Let $(A, \mathcal A, \nu)$ be a finite measure space. A function $f:A\rightarrow L$ is called \emph{simple} if there exist $x_1, x_2, ..., x_n\in L$ and $A_1, A_2, ..., A_n\in\mathcal A$ such that $f = \sum_{i = 1}^n\chi_{A_i}x_i$ where $\chi_{A_i}(a) = 1$ if $a\in A_i$ and $\chi_{A_i}(a) = 0$ otherwise. A function $f:A\rightarrow L$ is called $\nu$-\emph{measurable} or \emph{strongly measurable} if there exists a sequence of simple functions $(f_n)
$ with $\text{lim}_n\|f_n(a) - f(a)\| = 0$ a.e. A function $f:A\rightarrow L$ is called \textit{weakly measurable} or \textit{$L^*$-measurable} if for each $p\in L^*$ the functional $p\cdot f: A\rightarrow \mathbb R$ is measurable.  More generally, for $L'\subset L^*$ $f$ is called $L'$-measurable if $p\cdot f: A\rightarrow \mathbb R$ is measurable for every $p\in L'$.

\begin{lemma}\label{measurestrongmeasure} A function $f: A\rightarrow L$ is strongly measurable if and only if it is $\left(\mathcal A, \mathcal B\left (L^{\|\cdot\|}\right )\right )$-measurable and its values lie for $\nu$-almost all $a$ in a separable subset of $L$.
    
\end{lemma}

\begin{theorem}[\cite{diestel_vector_1977}, Pettis Measurability Theorem.]\label{pettismeasurability} $f:A\rightarrow L$ is strongly measurable if and only if

\begin{enumerate}
    \item $f$ is weakly measurable.
    \item There exists a $\nu$-null set $N\subset \mathcal A$ such that $f\left(A\setminus N\right)$ is norm separable.
\end{enumerate}
\end{theorem}

A strongly measurable function $f:A\rightarrow L$ is called \emph{Bochner integrable} if there exists a sequence of simple functions $(f_n)
$ such that $\text{lim}_n\int_A\|f_n (a)- f(a)\|d\nu = 0$ where the latter is an ordinary Lebesgue integral. Thus for each $ A'\in \mathcal A$, $\int_{ A'}fd\nu = \text{lim}_n\int_{ A'}f_n d\nu$. A useful characterization for Bochner integrability is the following: 

\begin{theorem}[\cite{diestel_vector_1977}, Theorem 2.]\label{bochnerintegral}
A strongly measurable function $f:A\rightarrow L$ is Bochner integrable if and only if $\int_A \|f(a)\|d\nu<\infty.$
\end{theorem}\label{simple}
\begin{remark}\label{simplefunctionsbochner} By Theorems \ref{pettismeasurability} and \ref{bochnerintegral}, simple functions are Bochner integrable.
\end{remark}

\begin{theorem}[\cite{diestel_vector_1977}, Theorem 3, \textit{Dominated Convergence Theorem}.]\label{DCT} Let $(f_n)$ be a sequence of Bochner integrable $F$-valued functions on $A$ for some $F\subset L$. If $\textnormal{lim}_n f_n = f$ in $\nu$-measure and if there exists a real valued Lebesgue integrable function $g$ on $A$ with $\|f_n\|\le g$, $\nu$-a.e., then $f$ is Bochner integrable and $\textnormal{lim}_n\int_Af_n d\nu = \int_A f d\nu$. In fact, $\textnormal{lim}_n\int_A\|f_n - f\|d\nu = 0$.
\end{theorem}

The following result is an adaptation of the uniform version of the Lebesgue Dominated Convergence Theorem.
\begin{theorem}[\cite{tausk_some_2001}, Lemma 3.2.]\label{UDCT} Let $J$ be and arbitrary set. Let $(f^j_n)_{j\in J, \ n\in\mathbb N}$ and $(f^j)_{i\in J}$ be a family of maps $f^j_n: A\rightarrow\mathbb R$, $f^j: A\rightarrow\mathbb R$ such that: (i) $f^j_n$ is measurable for all $n$ and $j\in J$, (ii) $\textnormal{lim}_n\underset{j\in J}{\textnormal{sup}}\left\{\left|f^j_n(a) - f^j(a)\right|\right\} = 0$ a.e. $a\in A$, and (iii) there exists an integrable map $h: A\rightarrow\mathbb R$ such that $|f^j_n(a)| \le h(a)$ for all $a\in A$, $n\in\mathbb N$ and $j\in J$. Then, all $f^j_n$ and $f^j$ are Lebesgue integrable and $\textnormal{lim}_n\underset{j\in J}{\textnormal{sup}}\left\{\left|\int_Af^j_n(a)d\nu(a) - \int_Af(a)d\nu(a)\right|\right\} = 0$.

\end{theorem}

We adapt a result from \cite{dunford_linear_1958} to our setting.
\begin{lemma}[\textbf{\cite{dunford_linear_1958}, Theorem 19 (c).}]\label{lebesgueintegral} Let $f:A\rightarrow L$ be a Bochner integrable function and $p$ belongs to $L^*$, then $p\cdot f$ is integrable and $p\cdot\int_A f(a)d\nu = \int_A p\cdot f(a)d\nu$ 
\end{lemma}

A function $f:A\rightarrow L^*$ is called \textit{weakly$^*$-measurable} or $L$-measurable (when $L$ is identified with its image under the natural embedding of $L$ into $L^{**}$) if for each $p\in L$ the functional $p\cdot f: A\rightarrow\mathbb R$ is measurable. 

 The weak$^*$-measurable function $f$ is \emph{Gel'fand integrable} over a set $ A'\in\mathcal A$ if there exists some $\int_{A'} f(a)d\nu\in L^*$ satisfying $p\cdot\left(\int_{A'} f(a)d\nu\right) = \int_{_{A'}}p\cdot f(a)d\nu$ for each $p \in L\subset L^{**}$. The vector $\int_{A'} f(a)d\nu$ is called the \emph{Gel'fand integral} of $f$ over $A'$. 
 
\begin{theorem}[\cite{aliprantis_infinite_2006}\label{Gel'fandintegral}, Corollary 11.53.]
If $f$ is weakly$^*$-measurable for all $p\in L\subset L^{**}$ and $\|f(a)\|^*\le s$ for some $s > 0$ and all $a\in A$, $f$ is Gel'fand integrable.
\end{theorem} 

We finish this section with the following result regarding the \textit{ change of the variable formula} for strongly measurable functions.

\begin{lemma}[\cite{dunford_linear_1958}, Lemma 8 (e).]
Let $A$ and $B$ be sets and let $f:A \rightarrow B$. Let $\mathcal B$ be a $\sigma$-algebra in $B$. Then $\mathcal A = \{f^{-1}(B'): B'\in\mathcal B\}$ is a $\sigma$-algebra in $A$. For a countably additive set function $\nu$ in $\mathcal B$, $\mu\circ f^{-1} = \nu$ defines a countably additive set function $\mu$ in $\mathcal A$. If $\nu$ is positive and a function $g:B\rightarrow L$ is strongly measurable and non-negative, then $$\int_{B'}g(b)d\nu = \int_{B'}g(b)d(\mu\circ f^{-1}) =  \int_{f^{-1}(B')}g\circ f(a)d\mu(a)$$

for every $B'$ in $\mathcal B$.

\end{lemma}

Taking into account that $\mathcal{A}$ is a sub-$\sigma$-algebra of $\mathcal{B}(A)$, so that whenever $(A, \mathcal{B}(A), \mu)$ is a measure space, $(A, \mathcal{A}, \mu)$ is also a measure space, we obtain the following corollary.

\begin{corollary}\label{chvariable}
Let $\mathcal{B}(A)$ and $\mathcal{B}(B)$ be the Borel $\sigma$-algebras of $A$ and $B$ respectively. Let $f : A \rightarrow B$ be $(\mathcal{B}(A), \mathcal{B}(B))$-measurable. If $\mu$ is a measure on $\mathcal{B}(A)$ and $\nu$ is a measure on $\mathcal{B}(B)$ such that $\mu \circ f^{-1} = \nu$, then for any strongly measurable and non-negative function $g : B \rightarrow L$, we have
\[
\int_{B} g(b) \, d\nu = \int_{B} g(b) \, d(\mu\circ f^{-1}) = \int_{A} (g \circ f)(a) \, d\mu.
\]
\end{corollary}

In Appendix \ref{measurability}, we shall state the conditions under which the above Lemma holds for weakly$^*$ measurable functions. Notice that weak$^*$-measurability is weaker than strong measurability when the Banach space is a dual space. Hence, Bochner integral implies that of Gel'fand integral. Furthermore, if $L$ is finite dimensional, the two of them coincide while they are equivalent to the Lebesgue integral when $L = \mathbb R^n$ for some $n$. 

\ 

A note on notation. Throughout the paper we use the symbol $\sigma$ in its two usual meanings: to denote a $\sigma$-algebra, and the weak topology $\sigma(L,L')$. This is a standard convention in measure theory and functional analysis, and its intended meaning should be clear from the context.

\section{ The basic of the model.}\label{sec:model}

\subsection{Discussion on the commodity-price pairing.}\label{}

\ 

Building on the previous section, we consider a Banach space $L$, its dual $L^*$, and the dual pairing $\langle L, L' \rangle$, where $L'\subseteq L^*$. We denote the weak topology on $L$ generated by $L'$ as $w$, defined by $\sigma(L, L')$, and the Mackey topology generated by $L'$ as $\tau$, defined by $\tau(L, L')$, which are, respectively, the weakest and strongest Hausdorff, locally convex topologies compatible with the duality $\langle L, L' \rangle$ (\cite{schaefer_topological_1999}, 1.2, p. 124 and Corollary 1., p. 131). In a similar way, we denote respectively by $w' = \sigma(L', L)$ and $\tau' = \tau(L', L)$ the weak and the Mackey topologies on $L'$. In addition to these topologies, we consider the norm topologies on $L$ and $L^*$ induced by the norms $\|\cdot\|$ and $\|\cdot\|^*$, respectively. If $L'\subset L^*$, then $\sigma(L, L')\subset \sigma(L, L^*)$ (\cite{schaefer_topological_1999}, Corollary, p. 125.) and, of course, $\sigma(L, L^*)\subset\tau(L, L^*)$. If $L' = L^*$, then $\tau(L, L^*)$ is the $\|\cdot\|$-topology (\cite{aliprantis_infinite_2006}, Corollary 6.23, p. 236).

A sequence (more precisely, a net) $(x_n)$ in $L$ $w$-converges to $x$ or converges weakly to $x$ if and only if $(p \cdot x_n)$ converges to $p \cdot x$ for every $p \in L'$, and $(x_n)$ in $L$ $\tau$-converges to $x$ if and only if for every $w'$-compact, convex and circled subset of $L'$, $G$, the sequence $\left ( \underset{p\in G}{\text{sup}}\left\{\left|p \cdot (x_n - x)\right|: p\in G\right\}\right )$ converges to $0$. That is to say, $p\cdot x_n\rightarrow p\cdot x$ uniformly for $p$ in $G$. In the same way, a sequence $(p_n)$ in $L'$ $w'$-converges to $p$ if and only if $(p_n \cdot x)$ converges to $p \cdot x$ for every $x \in L$, and $(p_n)$ in $L'$ $\tau'$-converges to $p$ if and only if for every $w$-compact, convex and circled subset of $L$, $C$, the sequence $\left (\underset{x\in C}{\text{sup}}\left\{\left|(p_n - p)\cdot x\right|: x\in C\right\}\right )$ converges to $0$. That is to say, $p_n \cdot x$ converges to $p \cdot x$ uniformly for $x$ in any $w$-compact subset of $L$. 

The subset $L_+ = \{x \in L \mid x \ge 0\}$ is the positive cone of $L$, which is $w$-closed. Let us consider a set $Y \subseteq L$. When $Y$ is endowed with the (relative) topology $w$ (respectively, with $\tau$ and $\|\cdot\|$), we denote it by $Y^w$ (respectively, $Y^\tau$ and $Y^{\|\cdot\|}$) unless it is understood. The $w$-closure (respectively, $\tau$ and $\|\cdot\|$-closure) of $Y$ will be denoted by $\overline{Y}^w$ (respectively, $\overline{Y}^\tau$ and $\overline{Y}^{\|\cdot\|}$). Notice that by the above reasons, $Y^\tau = Y^{\|\cdot\|}$ when $L' = L^*$. Whenever a topological space is Polish, every closed subset is also Polish for the induced topology. 

Our framework is general enough to include many pairs of commodity spaces and price simplex used in economic theory. First, if $L'$ is the predual of the commodity space $L$, then $L'$ may be seen as a proper subset of the topological dual $L^*$ of $L$ (\cite{aliprantis_infinite_2006}, Corollary 6.11 p. 231). The best-known cases in economic theory are those where $L = \mathscr L_{\infty}$, $\ell_{\infty}$, or the space of signed measures on a compact metric space $K$, $\mathcal{M}(K)$. These are nonseparable spaces. In these cases, $w = \sigma(L, L')$ is the weak$^*$-topology while $w' = \sigma(L', L)$ is the weak topology. The Mackey topology $\tau$ is coarser than the norm topology $\|\cdot\|$ on $L$ while $\tau' = \|\cdot\|'$ on $L'$ (\cite{aliprantis_infinite_2006}, Corollary 6.23, p. 236). Although prices should, strictly speaking, belong to $L^*$, in most cases they actually belong to $L'$. This is the case, for example, in commodity differentiation (see \cite{ostroy_nonatomic_1994}) where $L = \mathcal{M}(K)$ is endowed with the weak$^*$-topology $w = \sigma(\mathcal{M}(K), C(K))$ while the price set is compact for the norm of $C(K)$ or, equivalently, for the Mackey topology $\tau(C(K), \mathcal{M}(K))$. Another example is \cite{bewley_very_1991} where $\ell_{\infty}$ is endowed with the topology $\sigma(\ell_{\infty}, \ell_1)$ and the price simplex is $\sigma(\ell_1, \ell_{\infty})$-compact. Whenever in a Banach space weak and norm compactness coincide we say that this space has the Schur property. This is the case for $L' = \ell_1$ so the price simplex is $\|\cdot\|_1$-compact (and Polish). 
 Norm-bounded and closed subsets of $L$ are compact for the weak$^*$-topology $w = \sigma(L, L')$ by Alaoglu's Theorem. If, in addition, $L'$ is separable, such a set is metrizable (\cite{dunford_linear_1958}, Theorem 1, p. 426), whence $w$-Polish as it is $w$-separable and complete. Returning to our previous examples, for $L = \mathscr L_{\infty}$, $\ell_{\infty}$ or $\mathcal{M}(K)$, a norm-bounded subset $Q$ of $L$ is $\sigma(\mathscr L_\infty, \mathscr L_1)$-compact, $\sigma(\ell_\infty, \ell_1)$-compact, or $\sigma(\mathcal{M}(K), C(K))$-compact, respectively. Since $\mathscr L_1$, $\ell_1$, and $C(K)$ are separable spaces, $Q$ qualifies as a $w$-compact Polish space.

When ``enlarging" $L'$ as to be $L^*$ itself, $w = \sigma(L, L') = \sigma(L, L^*)$ is the weak topology and $w' = \sigma (L', L) = \sigma(L^*, L)$ is the weak$^*$-topology. Hence, prices are located in a weak$^*$-compact set. This is the case in \cite{bewley_existence_1972}, Section 3, where $w' = \sigma(ba, \mathscr L_{\infty})$. Notice that now both the Mackey and the norm topology coincide on $L$ but $\tau(L', L)$ is coarser than the topology induced by $\|\cdot\|'$. 

Secondly, if $L$ is a separable and reflexive space, then $w$ would be the weak topology $\sigma(L, L')$ where $L' = L^*$. Since $L$ is reflexive, it equals $L'' = L^{**}$, whereby the weak and weak$^*$-star topologies coincide. Consequently, by Alaoglu's Theorem, norm-bounded and closed subsets of $L$ are $w$-compact. Due to the separability of $L$, these subsets are metrizable and thus Polish (\cite{dunford_linear_1958}, Theorem 3, p. 434). Therefore, if $Q$ is a $\|\cdot\|$-bounded and closed subset of $L$, it forms a $w$-compact Polish space. If prices are situated in a $\|\cdot\|'$-bounded subset of $L'$, then similarly, they are $w'$-compact and Polish. Notice that Mackey and norm topologies coincide both in $L$ and in $L'$ by Corollary 6.23 in \cite{aliprantis_infinite_2006}. Typical examples belong to commodity spaces represented by $L =\mathscr L_p$ or $\ell_p$ and price spaces given by its dual $L' = L^* = \mathscr L_q$ or $\ell_q$ with $1 < p < \infty$, $1 < q < \infty$, and $\frac{1}{p} + \frac{1}{q} = 1$. In our notation, $w = \sigma(\mathscr L_p, \mathscr L_q)$ and $w' = \sigma(\mathscr L_q,\mathscr L_p)$ or $w = \sigma(\ell_p, \ell_q)$ and $w' = \sigma(\ell_q,\ell_p)$ respectively.
 
A third case arises when $L$ is a separable but not reflexive space, as discussed in \cite{khan_equilibria_1991}. In this scenario, the topologies $w$ and $\tau$ correspond to the weak and Mackey topologies on $L$ respectively, where $L' = L^*$. If prices are contained within a $\|\cdot\|'$-bounded subset of $L'$, then this subset is $w'$-compact and, as argued previously, metrizable. If the set $Q\subset L$ is $w$-compact, it is metrizable, $\|\cdot\|$-bounded (\cite{diestel_sequences_1984}, p. 17), and thus $w$-compact and Polish. Furthermore, the Mackey topology $\tau$ coincide with the norm topology $\|\cdot\|$ on $L$ but $\tau'$ is coarser than the $\|\cdot\|'$-topology on $L'$.
Beyond the cases of $C(K)$, $\ell_1$, and $\mathscr{L}_1$ discussed previously, examples of commodity spaces of this type include the space of null sequences $c_0$, as well as spaces generated by weakly compact subsets of $\mathscr{L}_\infty$ and $\ell_\infty$ (see \cite{cea_continuity_2024}).

In the following, within our general framework, we shall work under the duality $\langle L, L' \rangle$. As the preceding cases suggest, we assume that $L'$ is either the topological dual of $L$ or a predual of $L$. That is, we consider two alternatives: $L' = L^*$ or $L = (L')^*$, the latter implying $L' \subset L^*$. Furthermore, in view of the variety of admissible topologies, it is more natural to state results in a unified framework that encompasses all cases, rather than providing separate formulations for each specific topology. Accordingly, we shall write $\kappa \in \{\|\cdot\|, \tau, w\}$ for any of these topologies on $L$, and $\kappa' \in \{\|\cdot\|', \tau', w'\}$ for the corresponding topology on $L'$.

\subsection{Space of preferences.}\label{preferences}

\ 

Let $L_+$ be the positive cone of $L$ and let $\succ\subset L_+\times L_+$ be a transitive and irreflexive
binary relation on $L_+$ such that $\succ$ is $w$-open in $L_+\times L_+$. We use $x \succ y$ to denote $(x,y) \in \succ$ and $x \not\succ y$ to denote $(x,y) \notin \succ$. The set of all such relations is denoted by $\mathcal P_{L_+} $. For each $\succ \in \mathcal{P}_{L_+}$, we associate the set $P_\succ := \{(x, y) \in L_+\times L_+ : x \not\succ y \} = \left(L_+\times L_+\right)\setminus\succ$. Due to the continuity of $\succ$, each $P_\succ$ is a closed subset of $L_+\times L_+$. Let $\mathcal{C}_w(L_+\times L_+)$ be the set of all weakly closed subsets of $L_+\times L_+$. We denote by $\mathcal{T}^{\mathcal{C}}$ the \textit{topology of closed convergence} on $\mathcal{C}_w(L_+\times L_+)$. 

We embed $\mathcal{P}_{L_+}$ into $\mathcal{C}_w(L_+\times L_+)$ through the injection\footnote{Indeed, if $\succ \neq \succ'$ in $\mathcal{P}_{L_+}$ and $P_\succ = P_{\succ'}$, then $(L_+\times L_+) \setminus \succ = (L_+\times L_+) \setminus \succ'$, which implies $\succ = \succ'$, contradicting $\succ \neq \succ'$. Therefore, $P_\succ \neq P_{\succ'}$ whenever $\succ \neq \succ'$.} $g(\succ) = P_\succ$. Hence, the induced topology on $\mathcal{P}_{L_+}$ is $ \{ g^{-1}(U) : U \in \mathcal{T}^{\mathcal{C}} \}$. Therefore, a net $(\succ^\alpha)$ in $\mathcal P_{L_+}$ converges to $\succ$ in $\mathcal P_{L_+}$ if and only if $g(\succ^\alpha)$ converges to $g(\succ)$ in $\mathcal C_w(L_+\times L_+)$. 

Let $X$ be a common consumption set which is assumed to be a convex subset of $L_+$ containing the null vector and such that $X$ is $w$-compact Polish and $\|\cdot\|$-bounded. We point out that $X$ is $w$-closed. This is a general assumption concerning consumption set $X$. When $L' = L^*$, $w$-compactness of $X$ implies $\|\cdot\|$-closedness and boundedness (\cite{diestel_sequences_1984}, p. 17). Since $X$ is convex and $w$-separable, it is $\|\cdot\|$-separable (see Lemma \ref{separability}). On the other hand, if $L = (L')^*$, $\|\cdot\|$-boundedness implies $w$-compactness by Alaoglu's Theorem. This means that $X$ is $w$-closed whence, $\|\cdot\|$-closed.

Weakly compact consumption sets are found in various studies on large economies with infinitely many commodities such as \cite{khan_equilibria_1991}, \cite{noguchi_economies_1997a}, and \cite{noguchi_economies_1997b} among others. Furthermore, several authors have made use of common consumption sets. For example, \cite{bewley_very_1991} considers the non-separable space $\ell_{\infty}$ and assumes the existence of a common consumption set that is a weak$^*$-compact subset of $\ell^+_{\infty}$. More recent works like \cite{suzuki_competitive_2013}, \cite{suzuki_core_2013}, and \cite{khan_relaxed_2016} also assume a common consumption set that is weak$^*$- or weakly compact and metrizable.  In line with these approaches, we adopt $X$ as the consumption set across all economies. 

Since the consumption set is \(X \subset L_+\), we restrict attention to preferences defined on \(X\). Every \(\succ \in \mathcal{P}_{L_+}\) induces a strict preference relation on \(X\) via its restriction \(\succ \cap (X \times X)\). We define
\[
\mathcal{P}_X := \left\{ \succ_X \subset X \times X : \succ_X = \succ \cap (X \times X) \ \text{for some} \ \succ \in \mathcal{P}_{L_+} \right\}.
\]
We call \(\mathcal{P}_X\) the set of \textit{preference relations}. Since \(\succ \subset L_+ \times L_+\) is open, \(\succ_X\) is open in \(X \times X\). Let \(\mathcal{C}_w(X \times X)\) denote the collection of all closed subsets of \(X \times X\); that is, $\mathcal{C}_w(X \times X) = \left\{ \mathcal{O} \cap (X \times X) : \mathcal{O} \in \mathcal{C}_w(L_+ \times L_+) \right\}$. We equip \(\mathcal{C}_w(X \times X)\) with the relative topology induced by \(\mathcal{T}^{\mathcal{C}}\) denoted by \(\mathcal{T}_X^{\mathcal{C}}\). The set \(\mathcal{P}_X\) is naturally embedded into \(\mathcal{C}_w(X \times X)\) via $h(\succ_X) = \{(x,y) \in X \times X : x \not\succ_X y\} = (X \times X) \setminus \succ_X$. Observe that $h(\succ_X) = (X \times X) \cap P_\succ = (X \times X) \cap g(\succ)$
for some \(\succ \in \mathcal{P}_{L_+}\) such that \(\succ_X = \succ \cap (X \times X)\). Hence, the topology induced on \(\mathcal{P}_X\) is $\left\{ h^{-1}\bigl( U \cap \mathcal{C}_w(X \times X) \bigr) : U \in \mathcal{T}^{\mathcal{C}} \right\}$. 

Since \(X\) is \(w\)-compact and Polish, it follows from Corollary 3.95 in \cite{aliprantis_infinite_2006} that \(\bigl( \mathcal{C}_w(X \times X), \mathcal{T}_X^{\mathcal{C}} \bigr)\) is a compact metric space.

Beyond the set-theoretic and topological relationship between \(\mathcal{P}_{L_+}\) and \(\mathcal{P}_X\), there is also a measurable link between the two spaces.

\begin{lemma}\label{inducedpreference}
Consider the Borel \(\sigma\)-algebras on \(\mathcal{P}_{L_+}\) and \(\mathcal{P}_X\) generated by the topologies \(\mathcal{T}^{\mathcal{C}}\) and \(\mathcal{T}_X^{\mathcal{C}}\), respectively. The mapping \(R : \mathcal{P}_{L_+} \to \mathcal{P}_X\), defined by
\[
R(\succ) = \succ \cap (X \times X),
\]
is measurable.
\end{lemma}

Although the consumption space is typically taken to be \(L_+\), equilibrium allocations lie in a weakly compact subset of \(L_+\) under standard assumptions. Thus, when studying continuity properties of equilibrium allocations, assuming that \(X\) is \(w\)-compact entails no real loss of generality. In Section~\ref{applications}, we examine this point in more detail and show that continuity results can still be obtained when the consumption set is \(L_+\), relying on Lemma~\ref{inducedpreference}. For simplicity, and up to Section~\ref{applications}, we write \(\mathcal{P}\) instead of \(\mathcal{P}_X\), and denote its elements by \(\succ\).

\ 

\noindent\textbf{Assumption (PR).} For every $\succ \ \in \mathcal{P}$ and for every $x \in X$, either $x \in \overline{\{ x' \in X : x' \succ x \}}^w$ or $x$ is a satiation point for $\succ$.

 \ 
 
 When only the first part is considered, we say that preferences are $w$-locally nonsatiated on $X$. We also take into account the possibility of maximal elements in $X$ for $\succ$, since $X$ is weak-compact (see, for instance, \cite{podczeck_markets_1997}, \cite{lee_competitive_2013}, and \cite{khan_relaxed_2016}).

 \ 
 
 \textbf{Assumption (PC)} The space $\mathcal P$ is compact.
 
 \ 
 
Notice that since $\mathcal{C}_w(X \times X)$ is a compact metric space, the subset $\mathcal{P}$ is metrizable. If we exclude satiated bundles, it can be shown\footnote{Let us consider the set of locally non-satiated preferences $\mathcal{P}_{lns}\subset\mathcal P$. Through an adaptation of Proposition 4 (b) in \cite{grodal_note_1974}, p. 285, the set $\mathcal{P}_{lns}$ is a $\mathcal{G}_\delta$ set. By the classical Alexandroff lemma (see \cite{aliprantis_infinite_2006}, Lemma 3.34 on p. 88), we conclude that $\mathcal{P}_{lns}$ is completely metrizable. In addition, by Corollary 3.5 on p. 73 in \cite{aliprantis_infinite_2006}, we know that $\mathcal{P}_{lns}$ as a subset of the separable metric space $\mathcal{C}_w(X \times X)$, is separable.} that the set $\mathcal{P}$ is a compact subset of $\mathcal{C}_w(X \times X)$. Alternatively to Assumption (PC), one could take a compact subset of preferences satisfying Assumption (PR) and then define $\mathcal{P}$ as the closure of such a subset. Thus, $\mathcal{P}$ would be compact and metrizable. 

As a consequence, we have the following lemma:

 \begin{lemma}\label{contpref} For every $(\succ, x, y)\in \left (\mathcal P\times X\times X\right)$ such that $x\succ y$, there exists an open neighborhood $U_\succ$ of $\succ$ and $w$-open neighborhoods $V_x$ and $V_{y}$ of $x$ and $y$ respectively such that for all $(\succ', x', y')\in (U_\succ\cap\mathcal P)\times \left(V_x\cap X\right)\times \left(V_y\cap X\right)$, $x'\succ' y'$.
 \end{lemma}

The above lemma shows that the product topology on $\mathcal{P} \times X \times X$ allows us to formulate a uniform notion of continuity of preferences across economies. Lemma~\ref{contpref} guarantees that the preference relation behaves continuously in this product space, which will be fundamental for ensuring the stability of demand and equilibrium sets under perturbations.

\section{Economies.}\label{economies}

\ 

Let $\kappa \in \{\|\cdot\|, \tau, w\}$ be fixed, and let $\Omega \subset X$ be such that $(\Omega,\kappa)$ is compact and metrizable (hence a Polish space). This will be the initial-endowment space; for brevity write $\Omega^\kappa := (\Omega, \kappa)$. The characteristic set is $T = \mathcal P\times \Omega$. With a slight abuse of notation, $T^\kappa := \mathcal P \times \Omega^\kappa$, that is, the product space $\mathcal P \times \Omega$ endowed with the product topology $\mathcal T^C_X \times \kappa$. A Borel probability measure on $\Omega^\kappa$ ($T^\kappa$) is a probability measure on $\mathcal B(\Omega^\kappa)$ ($\mathcal B(T^\kappa)$) where $\mathcal B(\Omega^\kappa)$ ($\mathcal B(T^\kappa)$) is the Borel $\sigma$-algebra on $\Omega$ ($T$), that is, the $\sigma$ -algebra generated by the $\kappa$-open subsets of $\Omega$ ($T$). 

\begin{definition} A $\kappa$-\emph{economy} is a Borel probability distribution $\mu$ on the space of characteristics $(T, \mathcal B(T^\kappa))$. Hence, the set of $\kappa$-economies is $\mathcal M(T^\kappa)$, that is, the set of all Borel probability distributions on $T$ equipped with its Borel $\sigma$-algebra $\mathcal B(T^\kappa)$.
\end{definition}

Hereafter, we shall say $\kappa$\textit{-economy} or simply \textit{economy} when the context is clear.

Let $\iota: \Omega^\kappa \rightarrow \Omega^\kappa$ be the identity map. The total initial endowment for the economy $\mu \in \mathcal{M}(T^\kappa)$ is $\int_\Omega \iota  d\mu_\Omega$ where $\mu_\Omega$ is the distribution of endowments given by the marginal distribution of $\mu$ on the measurable space $(\Omega, \mathcal{B}(\Omega^\kappa))$. Since total initial endowment are defined by integrals, it is worth to distinguish between Gel'fand and Bochner economies: 

\begin{definition}
An economy $\mu$ is a Bochner economy (resp. a Gel'fand economy) if the total initial endowment $\int_\Omega \iota  d\mu_\Omega$ is a Bochner integral (resp. a Gel'fand integral). 
\end{definition}

Since Bochner integration applies to strongly measurable functions taking values in a Banach space, whereas Gel'fand integration applies to weak$^*$-measurable functions taking values in a dual space, we consider Bochner integrals in the context where $L' = L^*$ and Gel'fand integrals where $L = (L')^*$.

Fix $\kappa' \in \{\|\cdot\|', \tau', w'\}$ and let $S \subset L_+' \setminus \{0\}$ be such that $S^{\kappa'} := (S,\kappa')$ is compact and metrizable. We refer to $S^{\kappa'}$ as the price simplex.

The compactness of the sets of normalized prices and endowments is a standard assumption in equilibrium analysis. For endowments, it reflects the idea that agents' characteristics are not excessively heterogeneous. In existence problems involving infinitely many commodities, these sets remain compact, typically with respect to the weak topologies $w'$ and $w$ on prices and endowments, respectively. However, analysis of the continuity of equilibrium correspondences often requires additional topological properties (such as metrizability, completeness, or separability) which may not hold under weak topologies. This motivates the possible adoption of stronger topologies. In Section~\ref{applications}, we present examples of economic environments that satisfy these enhanced conditions.

We now consider the demand correspondence $D:T\times S\twoheadrightarrow X$ such that $$D(\succ, e, p) : = \{x\in X : p\cdot x\le p\cdot e \ \text{and} \ p\cdot x' > p\cdot e \ \text{if} \ x'\succ x\}$$In addition, for $p\in S$ let us consider the set $$ E_{p} =\{((\succ, e),\ x)\in T\times X :  x\in D(\succ, e, p)\}.$$ 

\begin{definition}\label{defeconomies}
A probability measure $\zeta\in\mathcal M(T^\kappa\times X)$ and a price vector $p\in L_+'$ is an equilibrium for the economy $\mu$ if:

\begin{enumerate}
    \item $\zeta_{T} = \mu$,
    \item $\int_X\iota d\zeta_{X} = \int_\Omega \iota d\mu_\Omega$  
    \item $\zeta(E_p ) = 1$.
    \end{enumerate}
    
\end{definition}
    
 A measure $\zeta\in\mathcal M(T^\kappa\times X)$ is an equilibrium allocation if there exist an economy $\mu\in \mathcal M(T^\kappa)$ and a price vector $p\in L_+'$ such that $(\zeta, p)$ is an equilibrium for $\mu$. Analogously, a price vector $p\in L_+'$ is an equilibrium price if there exist an economy $\mu\in \mathcal M(T^\kappa)$ and a probability measure $\zeta\in\mathcal M(T^\kappa\times X)$ such that $(\zeta, p)$ is an equilibrium for $\mu$.
    
 We point out that $\zeta_{X}$ denotes the distribution of the probability measure given by the marginal distribution of $\zeta$ on $X$ and $\zeta_{T}$ represents the distribution given by the marginal distribution of $\zeta$ on $T$. It is noted that $\int_X \iota d\zeta_{X}$ is a Bochner integral (respectively, a Gel'fand integral if $L$ is a dual space) by Corollary \ref{cBintegration} (respectively, Lemma \ref{Gintegration}). When the distribution $\zeta$ satisfies only Condition 1, it is called an \emph{allocation} for $\mu$. If it also satisfies Condition 2, we refer to it as a \emph{feasible allocation}. The next assumptions will be essential for achieving positive results regarding the objectives of this research.
   
 \ 
   
  \textbf{Assumption (IPVE)}(\emph{Individually Positively Valued Environments}). For all $(\mu, p)\in\mathcal M\left(T^\kappa\right)\times S$, $\mu_\Omega\left(\left\{e\in\Omega : p\cdot e > 0\right\}\right) = 1.$
 
  \ 

\textbf{Assumption (EP)}\label{EP}(Endowments-prices). Let $(e_n, p_n)$ be a sequence in $\Omega^{\kappa}\times S^{\kappa'}$ such that (i) $(p_n, e_n)\rightarrow (p, e)$ and (ii) $p_n\cdot e_n$ converges in $\mathbb R$ and (iii) each $p_n$ is an equilibrium price, then $\textnormal{lim}_n p_n\cdot e_n\ge p\cdot e$ if $\kappa = w$ and $\kappa' = w'$.

\ 
    
\begin{remark}\label{jointcont}\leavevmode
\vspace{-\baselineskip}

\ 

  \begin{enumerate}

 \item A sufficient condition for Assumption (IPVE) to hold, is that the set of initial endowments is bounded below by a quasi-interior point. This criterion aligns, for instance, with the assumptions made by Bewley \cite{bewley_very_1991} and Suzuki \cite{suzuki_competitive_2013} in their economic models. In cases where quasi-interior points do not exist in the positive cone (e.g. $\mathcal{M}(K)_+$ with $K$ uncountable), alternative conditions are in place to ensure that prices remain bounded away from zero, thereby maintaining (IPVE) (See the case in Section \ref{GEc}, 1.).

Let $p\in S$ and let $\mu_1$ and $\mu_2$ two economies. Let $\Omega_1$ the subset of $\Omega$ such that $p\cdot e = 0$ for every $e\in\Omega_1$. By (IPVE), $\mu_1\left(\mathcal P\times \Omega_1\right) = 0$. Let us define $\Omega_2$ in an analogous way for the economy $\mu_2$ so $\mu_2\left(\mathcal P\times\Omega_2\right) = 0$. Clearly, $\Omega_1$ does not necessarily coincides with $\Omega_2$ so the null-valued set of every economy (given a vector price $p$) is not common among the economies. Therefore, (IPVE) is clearly weaker than asking $p\cdot e > 0$ for every $e\in \Omega$.
    
\item Assumption (EP) is weaker than the assumption that the bilinear mapping $(e, p)\mapsto p\cdot e$ is $w\times w'$-lower semi-continuous\footnote{This assumption resembles Assumption (PR)(iii) in \cite{bonnisseau_existence_1999}, formulated for economies with increasing returns and infinitely many commodities.}. We only assume it for the weak topologies because in any of the other cases the mapping $(e, p)\mapsto p\cdot e$ on $\Omega^\kappa\times S^{\kappa'}$ is continuous. Indeed, when $\kappa = \|\cdot\|$ and $S$ is $\|\cdot\|'$-bounded\footnote{This is the case in \cite{cea_continuity_2024}.}, or when\footnote{Recall that $X$ is $\|\cdot\|$-bounded.} $\kappa' = \|\cdot\|'$, we have $\textnormal{lim}_n p_n\cdot e_n = p \cdot e$ (\cite{aliprantis_infinite_2006}, Corollaries~6.40 and~41, pp.~242--43.). 

More generally, we can argue that in order to ensure that $\textnormal{lim}_n p_n\cdot e_n = p \cdot e$, it is necessary that $\kappa$ be at least as strong as $\tau$, or $\kappa'$ be at least as strong as $\tau'$. As we can see, if $e_n \overset{\tau}{\rightarrow} e$ in $\Omega$ and $p_n \overset{w'}{\rightarrow} p$ in $S$, then $|p^n \cdot e^n - p \cdot e|$ tends to $0$ since
\[
|p^n \cdot e^n - p \cdot e| \le |p^n \cdot (e^n - e)| + |(p^n - p) \cdot e| \le \underset{q \in G}{\textnormal{sup}} \left\{ |q \cdot (e^n - e)| \right\} + |(p^n - p) \cdot e|,
\]
for $G$ a convex circled and $w'$-compact set containing $S$\footnote{For instance, $G$ could be the convex circled hull of $S$ which is $w'$-compact by the Krein--{\v{S}}mulian Theorem (\cite{aliprantis_infinite_2006}, 6.35, p.~241).}. The case in which $e_n \overset{w}{\rightarrow} e$ in $\Omega$ and $p_n \overset{\tau'}{\rightarrow} p$ in $S$ is analogous. Several of the economic applications presented in Section~\ref{applications} are based on one of these cases.
\end{enumerate}
\end{remark}

\section{Continuity of the equilibrium correspondence.}\label{continuity}

 To enable the comparison of economies and feasible allocations, we endow the spaces $\mathcal{M}(T^\kappa)$ and $\mathcal{M}(T^\kappa \times X)$ with their respective weak topologies. Notice that both $T^\kappa$ and $T^\kappa \times X$ are Polish spaces. Consequently, the set $\mathcal{M}(T^\kappa \times X)$ is complete and separable with respect to the Prohorov metric denoted by $\rho$\footnote{Strictly speaking, we should write $\rho$ to make explicit the dependence of the metric on the topology $\kappa$; however, for notational convenience we simply write $\rho$.} (\cite{billingsley_convergence_1999}, p.73) which, in turn, is equivalent to the weak topology on $\mathcal{M}(T^\kappa \times X)$. Of course, the projection on $T^\kappa$, denoted $\rho_{T}$, is equivalent to the relative weak topology on this space. Therefore, the separable and complete metric space $\left(\mathcal{M}(T^\kappa), \rho_{T}\right)$ is the space of economies. It is important to note that both $\left(\mathcal{M}(T^\kappa \times X), \rho\right)$ and $\left(\mathcal{M}(T^\kappa), \rho_{T}\right)$ are compact (\cite{hildenbrand_core_1974}, Eq. 30, p. 49). Therefore, the product space $\mathcal M(T^\kappa\times X)\times S$ endowed with the product topology $\rho\times\kappa'$ is compact and Polish.

\begin{definition}\label{dequilibria}
The distributional equilibrium correspondence
\[
DE\colon \left(\mathcal{M}(T^\kappa), \rho_T\right) \twoheadrightarrow  
\left(\mathcal{M}(T^\kappa \times X) \times S, \rho \times \kappa'\right)
\]
assigns to each economy \( \mu \in \mathcal{M}(T^\kappa) \) its corresponding set of equilibria \( DE(\mu) \) in \( \mathcal{M}(T^\kappa \times X) \times S \).
 \end{definition}

\begin{proposition}\label{uhc} The correspondence $DE$ is upper hemi-continuous if it satisfies Assumptions \textnormal{(PR)}, \textnormal{(PC)}, \textnormal{(EP)}, \textnormal{(IPVE)} and
\end{proposition}

\ 

\textbf{Assumption (SB) (satiated bundles)}. Let $(\succ, e)\in B$ for some $B\in\mathcal B(T^\kappa )$. If $x\in X$ is satiated for $\succ$, then $x\ge e$.

\ 

This assumption simply states that whenever a commodity bundle is satiated, it must be bigger than the corresponding endowment. In economies with a measure space of agents, this condition has become standard since Podczeck (1997) \cite{podczeck_markets_1997}\footnote{Besides \cite{podczeck_markets_1997}, see also \cite{lee_competitive_2013}, \cite{khan_relaxed_2016}, \cite{noguchi_economies_2000} and \cite{jang_equilibria_2020}.}.
    
We focus on the subspace of $\mathcal{M}(T^\kappa)$ consisting of economies that admit at least one equilibrium. Accordingly:
\begin{corollary}\label{baire} Let $\mathbb{E} := \left\{ \mu \in \mathcal{M}(T^\kappa) : DE(\mu) \neq \emptyset \right\}$.
Under the induced topology, the set $\mathbb{E}$ is a Baire space.
\end{corollary}

A straightforward application of Lemma \ref{essential}, given Proposition \ref{uhc} and Corollary \ref{baire}, yields the following.

\begin{lemma}
\label{lhc}
If $\mathbb E\neq\emptyset$, the correspondence $DE$ is lower hemi-continuous at every point of a dense residual subset $\mathbb E'$ of $\mathbb E$ under the assumption of Proposition \ref{uhc}.
\end{lemma}

Hence,
\begin{theorem}\label{contDW} Suppose that $\mathbb E\neq\emptyset$. Under Assumptions \textnormal{(PR)}, \textnormal{(PC)}, \textnormal{(EP)}, \textnormal{(IPVE)} and  \textnormal{(SB)}, there exists a dense residual subset $\mathbb E'$ of $\mathbb E$ such that the equilibrium correspondence $DE : \mathbb E'\twoheadrightarrow\mathcal{M}(T^\kappa\times X)\times S $ is continuous.  \end{theorem}

This is the main result of the present paper, built upon the preceding analyses. We refer to the economies in $\mathbb{E}'$ as having \emph{continuous equilibria}.

\begin{remark}    
We are aware that the preceding analysis is vacuous if $\mathbb{E} = \emptyset$. Although the main focus of our research is the continuity of equilibria rather than their existence, we cannot overlook the fact that a proper understanding of such a property requires, first of all, that equilibria do exist. Even more, Assumption (EP) explicitly presumes the existence of equilibria in neighboring economies. We address this important issue in the final section. 
\end{remark}

\section{The case of individualized economies.}\label{individual}

\ 

It is well known that every distributional economy $\mu$ has a representation through measurable functions. Indeed, by Skorokhod's theorem, for $\mu\in\mathcal M(T^\kappa)$, there exists a measurable function $\mathcal E$ and a measure space $(A, \mathcal A, \nu)$, such that $\mathcal E: (A, \mathcal A)\rightarrow (T, \mathcal B (T^\kappa))$ and $\nu\circ\mathcal E^{-1} = \mu$. Since $T^\kappa$ is complete, $(A, \mathcal A, \nu)$ may be chosen to be $(I, \mathcal B(I), \lambda)$, the unit interval with the Lebesgue measure $\lambda$ (\cite{hildenbrand_core_1974}, pp. 50-51). Furthermore, \cite{keisler_why_2009} Lemma 2.1 (ii), states that for any atomless measure space $(A, \mathcal A, \nu)$ and a distribution $\mu$ in $\mathcal M(T^\kappa)$, there exists a $(\mathcal A, \mathcal B(T^\kappa))$-measurable mapping $\mathcal E$ from $A$ into $T$ such that $\nu\circ\mathcal E^{-1} = \mu$. In any of these cases, $\mathcal E: (A, \mathcal A)\rightarrow (T, \mathcal B (T^\kappa))$ is said to be an individualized representation of $\mu$ if $\|\int_A e_\mathcal{E}(a)d\nu(a) \|<\infty$, where $e_\mathcal{E} = \text{Proj}_{\Omega}\mathcal E$ is the initial endowment and $(A, \mathcal A, \nu)$ represents an agent space. Note that there may be more than one representation. 

By invoking Corollary \ref{cBintegration}, $\int_A e(a)\, d\nu(a)$ is a Bochner integral. This is a Gel'fand integral when $L$ is a dual space due to Lemma \ref{Gintegration}. Furthermore, every measurable function $x: (A, \mathcal{A}) \rightarrow (X, \mathcal{B}(X))$ is an \emph{allocation}. It is a \emph{feasible} one if $\int_A x(a)\, d\nu(a) \le \int_A e(a)d\nu(a)$, where $\int_A x(a)d\nu(a)$ exists as a Bochner or Gel'fand integral according to Corollary \ref{cBintegration} or Lemma \ref{Gintegration} respectively. The mapping $x: (A, \mathcal{A}) \rightarrow (X, \mathcal{B}(X))$ is an equilibrium allocation for $\mathcal E$ if there exists $p \in L_+'$ such that $p\cdot x(a) \le p\cdot e_\mathcal{E}(a)$ and $p\cdot v > p\cdot e_\mathcal{E}(a)$ whenever $v \succ_a x(a)$ for almost all $a \in A$. We remark that $\text{Proj}_{T} (\mathcal{E}, x) = \mathcal{E}$, $\text{Proj}_\Omega (\mathcal{E}, x) = \text{Proj}_\Omega \mathcal{E} = e_{\mathcal{E}}$, and $\text{Proj}_X (\mathcal{E}, x) = x$. Note that the measurable function $(\mathcal{E}, x): A \rightarrow T \times X$ is a representation of the allocation $\zeta$ for the economy $\mu$.

The following lemma shows the relationship between distributional and individualized equilibria.

\begin{lemma}\label{wequivalence}\

\begin{enumerate}
\item Let $(x, p)$ be an equilibrium of the economy $\mathcal E$ whose agent space is the measure space $(A, \mathcal A, \nu)$ and the charasteristics set is $(T, \mathcal B(T^\kappa))$ then $\nu\circ\left(\mathcal E, x\right)^{-1}$ is a distributional Walrasian equilibrium of the distributional economy $\nu\circ\mathcal E^{-1}$ in $\mathcal M\left(T^\kappa\right)$. 
\item Let $\zeta$ be a distributional equilibrium of the distributional economy $\mu\in \mathcal M\left(T^\kappa\right)$ and let $(A, \mathcal A, \nu)$ be any atomless finite measure space. There exist measurable mappings $\mathcal E': A\rightarrow T$ and $x': A\rightarrow X$ such that $\nu\circ\mathcal E'^{-1} = \mu$, $x'$ is an individualized equilibrium allocation for $\mathcal E'$ and $\nu\circ (\mathcal E', x')^{-1} = \zeta$.
\end{enumerate}
\end{lemma}

The proof of Point 1 is immediate. For Point 2, we rely on Lemma 5.4 in \cite{suzuki_fundamentals_2020}.
\ 

The above results together with Proposition \ref{uhc} imply that large individualized economies have the closed graph property when convergence is in distribution.

\begin{proposition}\label{wclosenessindiv}
Let $(A_n, \mathcal A_n, \nu_n)$ be a sequence of probability spaces denoting agent spaces. Let $(\mathcal E_n, x_n)$ be a sequence of economies and allocations such that for each $n\in\mathbb N$, $\mathcal E_n: (A_n, \mathcal A_n)\rightarrow (T, \mathcal B(T^\kappa))$ and $x_n: (A_n, \mathcal A_n)\rightarrow (X, \mathcal B(X))$ are measurable functions. Let $(p_n)\in S^{\kappa '}$ be a sequence of prices such that for each $n\in\mathbb N$, $(x_n, p_n)$ is an equilibrium for $\mathcal E_n$. Let $(A, \mathcal A, \nu)$ be an atomless probability space and let $\mathcal E : (A, \mathcal A)\rightarrow (T, \mathcal B(T^\kappa))$ and $x: (A, \mathcal A)\rightarrow (X, \mathcal B(X^\kappa))$ be measurable functions. If $(\mathcal E_n, x_n)$ converges in distribution to $(\mathcal E, x)$ and $(p_n)$ converges to $p$, then $(x, p)$ is an equilibrium for $\mathcal E$ under the assumptions of Proposition \ref{uhc}.
\end{proposition}

By taking $(A_n, \mathcal{A}_n, \nu_n) = (A_n, 2^{|A_n|}, \nu_n)$ where $A_n$ is finite, the above result shows that an individualized economy with an atomless probability measure space of agents can be viewed as the limit of an increasing sequence of economies with finite agents.

Let $\mathcal{E}: A \rightarrow T$ be an individualized economy such that $\zeta \in DE(\nu \circ \mathcal{E}^{-1})$. Notice that Lemma \ref{wequivalence} does not address the existence of a measurable mapping $x: A \rightarrow X$ such that $\nu \circ (\mathcal{E}, x)^{-1} = \zeta$. Several studies have explored this issue, notably \cite{keisler_why_2009}, which underscores the need to consider \emph{saturated} atomless probability spaces. The space $(A, \mathcal{A}, \nu)$ satisfies the \emph{saturation} property if for every $\zeta \in \mathcal{M}(T^\kappa \times X)$ such that for each measurable function $r: A \rightarrow T$ with $\zeta_{T} = \nu \circ r^{-1}$, there exists a measurable function $s: A \rightarrow X$ such that $\nu \circ (r, s)^{-1} = \zeta$. The requirement of saturation in agent spaces restricts the set of admissible spaces. For instance, one cannot use the Lebesgue unit interval, which is widely used in economic theory. By counterpart, saturation allows for a more general result than Proposition \ref{wclosenessindiv}, in the sense that it is not necessary to assume the convergence of a distributional allocation to a specific allocation.

\begin{theorem}\label{closenessindiv}
Let $(A_n, \mathcal A_n, \nu_n)$ be a sequence of probability spaces denoting agent spaces. Let $(\mathcal E_n, x_n)$ be a sequence of economies and allocations such that for each $n\in\mathbb N$, $\mathcal E_n: (A_n, \mathcal A_n)\rightarrow (T, \mathcal B(T^\kappa))$ and $x_n: (A_n, \mathcal A_n)\rightarrow (X, \mathcal B(X))$ are measurable functions. Let $(p_n)\in S^{\kappa '}$ be a sequence of prices such that for each $n\in\mathbb N$, $(x_n, p_n)$ is an equilibrium for $\mathcal E_n$. Let $(A, \mathcal A, \nu)$ be an atomless saturated probability space and let $\mathcal E : (A, \mathcal A)\rightarrow (T, \mathcal B(T^\kappa))$ be a measurable function. If $(\nu_n\circ\mathcal E_n^{-1})$ converges weakly to $\nu\circ\mathcal E^{-1}$, $(\nu_n\circ x^{-1}_n)$ converges to some measure $\gamma\in\mathcal M\left(X\right)$ and $(p_n)$ converges to $p$, there exists a measurable mapping $x: (A, \mathcal A)\rightarrow (X, \mathcal B(X))$ such that $(x, p)$ is an equilibrium for $\mathcal E$ and $\nu\circ x^{-1} = \gamma$ under the assumptions of Proposition \ref{uhc}.
\end{theorem}

\section{Applications.}\label{applications}

\ 

In this section, we apply the main results to the problem of essential stability in some of the most widely studied economic models. We emphasize that our continuity results are non-trivial, provided the existence of equilibrium is guaranteed. As we shall illustrate, certain existence theorems define economies with consumption sets equal to $L_+$, so that the admissible preference relations are of the form $\succ\subset L_+\times L_+$. At first glance, this may hinder the direct application of Theorem \ref{contDW} to such economies. However, whenever the equilibrium is supported on a weakly compact subset $X \subset L_+$, the economy induces a new one, preserving the same equilibrium, but with consumption set $X$ and preference relation on it, thereby satisfying the assumptions of Theorem \ref{contDW}. 

By Lemma \ref{inducedpreference} we get:

 \begin{proposition}\label{inducedequilibrium} Let $(A, \mathcal A, \nu)$ be a finite measure space and let $\mathcal E: (A, \mathcal A)\rightarrow \left(\mathcal P_{L_+} \times \Omega, \mathcal B(P_{L_+} \times \Omega^\kappa)\right)$ be an individualized economy. Let $x: (A, \mathcal A) \rightarrow (X\times\mathcal B(X))$ be a feasible allocation such that $X \subset L_+$ is $w$-compact and Polish, and let $p \in L'_+$. If $(x, p)$ is an equilibrium for $\mathcal E$, then there exists and economy $\mathcal E_X: (A, \mathcal A) \rightarrow \left(\mathcal P_X \times \Omega^\kappa, \mathcal B(\mathcal P_X \times \Omega^\kappa)\right)$ such that $(x, p)$ is also an equilibrium for $\mathcal E_X$.
\end{proposition}

The purpose of Proposition \ref{inducedequilibrium} is to show that, even when an economy $\mu$ has support in $\mathcal{P}_{L_+} \times \Omega$, if its equilibrium allocation lies in $\mathcal{M}\left(\mathcal{P}_{L_+} \times \Omega^\kappa \times X\right)$ for some weakly compact subset $X \subset L_+$ (as often occurs), then the set $\mathbb{E} \subset \mathcal{M}\left(\mathcal{P}_X \times \Omega^\kappa\right)$ is nonempty. This justifies studying the continuity of equilibria in $\mathcal{M}\left(\mathcal{P}_X \times \Omega^\kappa \times X\right)$ for economies supported in $\mathcal{P}_{L_+} \times \Omega^\kappa$.

\subsection{Spaces of Bochner economies.}\label{Bec}

\begin{enumerate}

 \item \textbf{Neoclassical economies}. Let $L_+ = C(K)_+$, where $K$ is a compact subset of $\mathbb R^{\ell}$. Consider $X \subset C(K)_+$ as a convex and weakly-compact set that includes $0$ as a lower bound, $\hat{x}$ as an upper bound, and $\xi \in \text{int} C(K)_+$. We assume that for all $x \in X$ such that $x \neq \hat{x}$, $\hat{x} - x \in \text{int} C(K)_+$\footnote{See the example in \cite{noguchi_economies_2000}}. Moreover, $X^{\|\cdot\|}$ is $\|\cdot\|$-Polish, where $\|\cdot\|$ denotes the sup norm $\|\cdot\|_\infty$, because $X$ is convex and $\left(C(K), \|\cdot\|\right)$ is Polish. Let us consider the set of $w$-continuous, strictly increasing, and concave real functions on $X$, representing neoclassical preferences and let $\mathcal{P}_X$ be a compact subset of that. Additionally, let $\Omega$ be a norm-compact subset of $X$ bounded below by $\alpha \xi$, where $\alpha > 0$ is a parameter close to zero. Price set is $S = \{p\in L_+^* : p\cdot\xi = 1\}$ which is $w' = \sigma(C(K)^*, C(K))$-compact (\cite{jameson_ordered_1970}, Theorem 3.8.6, p. 123) and metrizable since $L$ is separable (\cite{dunford_linear_1958}, Theorem V.1, p. 426). Therefore, the set of economies \( \mathcal{M}(T^{\|\cdot\|}) \) consists of Bochner economies, and together with the price simplex \( S^{w'} \), they satisfy condition (EP).

By invoking the Main Theorem in \cite{noguchi_economies_2000} and combining it with Lemma \ref{wequivalence}(1), we conclude that $\mathbb{E} \subset \mathcal{M}\left(T^{\|\cdot\|}\right)$ is nonempty. Therefore, $DE: \mathbb{E} \twoheadrightarrow \mathcal{M}\left(T^{\|\cdot\|} \times X\right)\times S$ is a nonempty, $\left(\rho_{T}, \rho\times w'\right)$-upper hemi-continuous correspondence by Proposition \ref{uhc}. Hence, equilibria in a dense residual subset $\mathbb{E}'$ of $\mathbb{E}$ are $\left(\rho_{T}, \rho\times w'\right)$-continuous (Theorem \ref{contDW}). For further insights, we refer to Example 7.3 in \cite{cea_continuity_2024} regarding this space of economies.

\ 

\item \textbf{Neoclassical economies with separable utilities}. 

Consider neoclassical economies with the same consumption set as in the previous item. Let us consider norm separable utilities, so $\mathcal{P}_X$ is a compact space of $w$-continuous, strictly increasing, strictly quasi-concave, and separable utilities on $X$. The set of initial endowments $\Omega$ is now a $w$-compact subset of $X$ (for example, $\Omega = \{x \in X : x \ge \alpha \xi\}$ where $\alpha > 0$ is the same as previously). The characteristic set is $T^w = \mathcal{P}_X \times \Omega^w$, making every probability distribution on it a neoclassical exchange Bochner ($w$-)economy.

Prices should lie in the positive cone of the dual space of $C(K)$. However, with separable utilities, the set of prices that support equilibria belong to a $\|\cdot\|'$-compact subset of it (see \cite{covarrubias_number_2013}). Hence, these economies satisfy Assumptions (PR), (PC), (IPVE), (EP) and (SB). Consequently, all continuity results from Section \ref{continuity} apply whenever $\mathbb{E} \neq \emptyset$.

\ 

  \item \textbf{Allocations of non-renewable resources over an infinite time horizon (I)}.
 
 A suitable space for modeling non-renewable resources over time is that of null sequences $c_0 = \{(y_n) : \text{lim}_n y_n = 0\}$, since these resources tend to zero in the long run. The space $c_0$ under the sup norm is a separable Banach lattice with an order continuous norm (\cite{aliprantis_infinite_2006}, p. 529; \cite{godefroy_banach_2001}, p. 2). The topological dual of $c_0$ is $\ell_1$.

Consider the space of preferences $\mathcal P_{c_0^+}$. Choose $\xi \in c_0$ such that $\xi_i > 0$ for each $i$. Let $\Omega$ be an order interval bounded below by $\xi$. Hence, $\Omega$ is $w = \sigma(c_0, \ell_1)$-compact given the symmetry of $\langle c_0, \ell_1\rangle$ (\cite{aliprantis_infinite_2006}, Theorem 8.60 (3), p. 344) and metrizable by \cite{dunford_linear_1958}, Theorem 2, p. 426. Under further assumptions, Theorem 6 in \cite{zame_markets_1986} together with Lemma \ref{wequivalence}(1) allow us to assert that there exist economies in the set $\mathcal{M}\left(\mathcal P_{c_0^+} \times \Omega^w\right)$ that have equilibria in $\mathcal{M}\left(\left(\mathcal P_{c_0^+} \times \Omega^w\right) \times X\right)\times S$, where $X$ is a $w$-compact subset of $c_0^+$  and $S$ is the $w' = \sigma(\ell_1, \ell_\infty)$-compact subset of $\ell_1$ given by an specific order interval (see Proof of Theorem 6 on pp. 10-22 - 10-26.). Actually, $S$ is $\|\cdot\|_1$-compact since $\ell_1$ has the Schur property. The set $X$ is metrizable for the same reason as $\Omega$.

Take the set of preferences $\succ_X$ that are the preorders on $X$ induced by $ \succ$ on $c_0^+$. Let $\mathcal{P}_X$ be a compact subset of such as preorders satisfying (PR). By Proposition \ref{inducedequilibrium} we infer the existence of $\mu \in \mathcal{M}(\mathcal P_X\times\Omega^{w})$ with corresponding equilibria in $\mathcal{M}(\mathcal P_X\times\Omega^{w}\times X)\times S$ so $\mathbb E\neq\emptyset$. By assuming condition (SB) (by the choice of $\Omega$, Assumption (IPVE) holds), the correspondence $DE: \mathbb{E} \twoheadrightarrow \mathcal{M}(T^{w} \times X)\times S$ is $\left(\rho_T, \rho\times w'\right)$-upper hemi-continuous by Proposition \ref{uhc} and continuous in a dense residual subset of $\mathbb{E}$ by Theorem \ref{contDW}.

\begin{example}
 Let $b > > 0$ and let be the set $\mathcal P_{[0, b]}\times \Omega = \left([0, b], \left(u_i\right)_{i = 1}^n\right)\times [\xi, b]$ for $0 <<\xi <  b$ and utility functions $u_i$ on $c_0^+$ restricted to $ X = [0, b]$ that are $w$-continuous, strictly increasing and concave. We note that both $ X$ and $\Omega$ are $w = \sigma(c_0, \ell_1)$-compact. Prices belong to an order interval $[a\gamma, \ b\gamma]$ where $\gamma$ is a strictly positive vector in $\ell_1^+$ and $a$ is a strictly positive real number. We note that such an order interval is $\|\cdot\|$-compact. One easily checks that $T^w = \mathcal P_{[0, b]}\times \Omega^w$ satisfies \textnormal{(EP)}.
    
 By considering the assumptions of Theorem 1 of \cite{zame_markets_1986}, we may deduce the existence of a nonempty subset $\mathbb E$ of $\mathcal M\left( T^w\right)$ where each economy $\mu $ has equilibria in $\mathcal M\left( T^w \times X\right)\times S$. Assumptions \textnormal{(PR)}, \textnormal{(PC)}, \textnormal{(EP)}, \textnormal{(IPVE)} and \textnormal{(SB)} hold. Then, there is a dense subset of $\mathbb E$ whose equilibria are continuous. 
\end{example}

  \ 
  
 \item \textbf{Allocations of non-renewable resources over an infinite time horizon (II).}

We retain the same model as before, but now consider an endowment set that is $\|\cdot\| = \|\cdot\|_\infty$-compact and a price simplex $S = \left\{p\in\ell_1:\|p\|_1\le 1\right\}$ that is weak$^*$-compact by Alaoglu's Theorem since $\ell_1 = c_0^*$. Therefore, the economies in $\mathcal{M}\left(T^{\|\cdot\|}\right)$ are Bochner economies, and the set satisfies Assumption (EP). 

 \begin{example}
Consider the characteristic set $T = \left(([0, b], \left(u_i\right)_{i = 1}^n\right)\times \Omega^{\|\cdot\|}$ for $\Omega = \left\{x\in\mathbb R^{\mathbb N}: a_n\le x_n \le b_n\right\}$ given $(a_n)_{n\in\mathbb N}$ and $(b_n)_{n\in\mathbb N}$ in $\mathbb R^{\mathbb N}$. The latter is a compact subset of $(c_0, \|\cdot\|_\infty)$. 
\end{example}

\ 

 \item \textbf{Allocations of non-renewable resources over an infinite time horizon (III).}

Let $\ell_1$ denote the space of positive and summable sequences, which is suitable for modeling non-renewable resources such as oil where the total sum available over time is bounded. The norm of any $x \in \ell_1$ is $\|x\|_1 = \sum_{n=1}^\infty |x_n| < \infty$. The space $\ell_1$ is a separable Banach lattice with order continuum norm (\cite{aliprantis_infinite_2006}, Theorem 16.25, p. 537). The dual space of $\ell_1$ is $\ell_\infty$, and the weak topology on $\ell_1$ is denoted by $w = \sigma(\ell_1, \ell_\infty)$.

Consider economies in $\mathcal{M}\left(\mathcal P_{\ell_1^+} \times \Omega\right)$ where $\Omega = [a, b]$ with $a$ and $b$ strictly positive vectors in $\ell_1^+$ such that the latter is an upper bound for total initial endowments. Under suitable conditions, Theorem 6 of \cite{zame_markets_1986}\footnote{We may use also Theorem 4 in the cited paper.} and Lemma~\ref{wequivalence}~(1) imply that there exist probability measures on $\left(\mathcal P_{\ell_1^+} \times \Omega\right)$ having equilibria in $\mathcal{M}\left(\mathcal P_{\ell_1^+} \times \Omega \times X\right)$, where $X $ is an order interval in $\ell_1^+$, bounded below by $0$ and $\Omega\subset X$. Since $\langle \ell_1, \ell_\infty \rangle$ is symmetric, $X$ is $w = \sigma(\ell_1, \ell_\infty)$-compact and according to \cite{aliprantis_infinite_2006}, Theorem 16.24, p. 537, it is $\|\cdot\| = \|\cdot\|_1$-compact and hence $\|\cdot\|_1$-Polish.\footnote{The relatively compactness of the set $X$ is discussed in pp. 6-1 and 6-2 of \cite{zame_markets_1986}.} Consequently, the same applies for $\Omega$.

Without loss of generality, let $\mathcal P_X$ be the set of preferences on $X$ induced by $\succ$ in $\mathcal P_{\ell_1^+}$ which is assumed compact and satisfying (PR). The set of characteristics is then $T^{\|\cdot\|} = \mathcal{P}_X \times \Omega^{\|\cdot\|}$, where economies $\mu$ in $\mathcal{M}\left(T^{\|\cdot\|}\right)$ are Bochner economies.

The set of prices is an order interval in \( S^\infty = \{ p \in \ell_\infty^+ : \|p\|_\infty \le 1 \} \), which is \( w' = \sigma(\ell_\infty, \ell_1) \)-compact by the Alaoglu Theorem. Since \( \ell_1 \) is separable, \( S^\infty \) is metrizable under \( w' \).

From Proposition \ref{inducedequilibrium}, we deduce that the set of equilibria \( \mathbb{E} \subset \mathcal{M}\left(T^{\|\cdot\|}\right) \) is nonempty. Condition (IPVE) holds because \( \Omega \) is bounded below by a quasi-interior point of \( \ell_1^+ \). All continuity results hold in a dense subset of \( \mathbb{E} \) if condition (SB) holds.

\begin{example}
    Let $X = [0, b]\subset \ell_1^+$ be the common consumption set for $b = (b_n) > 0$. Let $\Omega = [a, b]$ for $0 < a < b$, where $a = (a_n)$ is a quasi-interior point which contains vectors such as $e = a + \left((b_n - a_n)n^{-1}\right)$. The space of preferences is $\mathcal P_{[0, b]} = \left([0, b], \textnormal{co}\left\{u_1, u_2\right\}\right)$ for $u_1(x) = \sum_{n = 1}^\infty x_n = \|x\|_1$ and $u_2(x) = \sqrt{\|x\|_1}$. We note that both functions on $X$ are $\sigma(\ell_1, \ell_\infty)$-continuous, strictly increasing and concave. Hence, we deduce that assumptions \textnormal{(PR)} and \textnormal{(PC)} hold. Assumption \textnormal{(IPVE)} is also satisfied because $p\cdot a > 0$ for every $p\in S^\infty$. The only satiated consumption bundle is $b$, whence we deduce that Assumption \textnormal{(SB)} also holds. If $\mathbb E\neq\emptyset$, $DE$ is continuous on a dense residual subset of $\mathbb E$.
\end{example}

\ 

\item\textbf{Reflexive economies}

Examples of reflexive Banach spaces include $\mathscr L_p(M, \mathscr M, \nu) $, where $ (M, \mathscr M, \nu) $ is a measure space. In fact, these spaces are Banach lattices according to the Riesz-Fischer theorem (\cite{aliprantis_infinite_2006}, p. 464) for $ 1 \leq p \leq \infty $. For $ 1 \leq p < \infty $, $\mathscr L_p(M, \mathscr M, \nu) $ is separable (\cite{luenberger_optimization_1969}, Example 4, p. 43), and for $ 1 < p < \infty $, it is also reflexive (\cite{aliprantis_infinite_2006}, Corollary 13.27, p. 473).

Another example is the sequence spaces $ \ell_p $. These are Banach lattices (\cite{aliprantis_infinite_2006}, Theorem 16.17, p. 534) for $ 1 \leq p \leq \infty $, separable for $ 1 \leq p < \infty $ (\cite{aliprantis_infinite_2006}, Theorem 16.19, p. 535), and reflexive for $ 1 < p < \infty $ (\cite{aliprantis_infinite_2006}, Corollary 16.21, p. 537).

Both, $\mathscr L_p$ and $\ell_p$, are widely used in economic theory. Depending on the context, they allow to model uncertainty, financial shocks, non-renewable resources, etc. We start with the commodity space $\ell_p^+$ and consider the vectors $a$ and $b$ in $\ell_p^{++}$ such that $a < b$.  Consider the economies in the set $\mathcal{M}\left(\mathcal P_{\ell_p^{+}}\times \Omega\right)$ where $\Omega\subset [a, b]$ is $\|\cdot\| = \|\cdot\|_p$-closed. By appealing again to Theorem 6 in \cite{zame_markets_1986} and Lemma \ref{wequivalence} (1), we find (by considering further assumptions) that there exists a subset of the latter that has all their equilibria in the set $\mathcal{M}\left(\mathcal P_{\ell_p^{+}}\times \Omega\times X\right)\times S$ for $X = [0, b]\subset\ell_p^{+}$ and $S$ a $w'$-compact subset of the dual $\ell_q$ where $\frac{1}{p} + \frac{1}{q} = 1$. 

Let $\mathcal P_{X}$ as above, satisfying (PR) and (PC). We note that both $\Omega$ and $X$ are relatively $\|\cdot\|$-compact since they are subsets of order intervals and all order intervals in $\ell_p$ are $\|\cdot\|$-compact for $1 < p < \infty$. Notice that $\Omega^{\|\cdot\|}$ is Polish since $(\ell_p, {\|\cdot\|})$ is separable (see \cite{fristedt_modern_1997}, Proposition 3, p. 350) and thus, the set of characteristics $T^{\|\cdot\|} = \mathcal P_X\times\Omega^{\|\cdot\|}$ is compact polish and the economies in $\mathcal{M}\left(T^{\|\cdot\|}\right)$ are Bochner. 

Clearly, $S$ is $w'$-metrizable and the economies in $\mathcal{M}\left(T^{\|\cdot\|}\right)$ satisfy (EP). Let $\mathbb E$ be the subset of $\mathcal{M}\left(T^{\|\cdot\|}\right)$ having equilibria in $\mathcal{M}\left(\mathcal P_X\times \Omega\times X\right)\times S$. Due to the above remarks, we know that $\mathbb E$ is nonempty. Therefore, there exists a dense residual subset $\mathbb E'$ of $\mathbb E$ where the equilibrium correspondence $DE:\mathbb E\twoheadrightarrow \mathcal{M}\left(\mathcal P_X\times \Omega\times X\right)\times S$ is continuous under (SB).
\end{enumerate}

 Now we move on to the reflexive spaces $\mathscr L_p$ ($1 < p < \infty$). By repeating the same steps as in $\ell_p$, we conclude that there exist order intervals $[a, b]$ and $X = [0, b]$ in $\mathscr L^+_p$ such that there is a subset of $\mathcal M\left(\mathcal P_{\mathscr L_p^+}\times [a, b]\right)$ with equilibria in $\mathcal M\left(\mathcal P_{\mathscr L_p^+}\times [a, b]\times X\right)\times S$ for $S$ a $w' = \sigma (\mathscr L_q, \mathscr L_p)$-compact subset of $\mathscr L_q$.
 
 The sets $[a, b]$ and $X$ are compact with respect to the weak topology $w = \sigma(\mathscr L_p, \mathscr L_q)$, but not with respect to the norm topology $\|\cdot\| = \|\cdot\|_p$. This contrasts with the previous case. Consequently, in order to study the continuity of the equilibrium correspondence, we reduce the set of economies under consideration to a subset of $\mathcal M\left(\mathcal P_{X}\times  \Omega\right)$ where $\Omega$ is a $\|\cdot\|$-compact subset of $[a, b]$\footnote{The existence of such a set $\Omega$ may depend on the underlying measure space $(M, \mathscr M, \nu)$. For example, if $M = \mathbb R^n$ and $\mathscr M = \mathcal B(\mathbb R^n)$ so that $(M, \mathscr M, \nu)$ is the Lebesgue measure space on $\mathbb R^n$, $\Omega$ is compact if and only if it satisfies the properties of $\|\cdot\|$-boundedness, \textit{equicontinuity} and \textit{equivanishing} by Kolmogorov-Riesz compactness Theorem. More generally, if $M$ is a locally compact Hausdoff group satisfying these conditions, $\Omega$ is compact for the norm topology by Weil Theorem. We refer the reader to \cite{krukowski_notes_2023} for a comprehensive survey of these results, including relevant extensions and precise definitions. A more economical discussion about the norm compactness of the set of initial endowments is given in \cite{zame_markets_1986}, Sections 6 and 7.} and $\mathcal{P}_{X}$ is defined as before throughout this section, whence satisfying (PR) and (PC) (condition (IPVE) also holds). Since $S$ is metrizable, the space of Bochner economies also satisfies Assumption \textnormal{(EP)}.
 
By the preceding considerations, there exists a set $\mathbb{E} \subset \mathcal{M}\left(\mathcal{P}_{X} \times \Omega\right)$ such that $\emptyset \neq DE(\mu) \subset \mathcal{M}\left(\mathcal{P}_{X}\times \Omega\right) \times S$ for every $\mu \in \mathbb{E}$. Moreover, by Theorem~\ref{contDW}, the correspondence $DE$ is $\left(\rho_T, \rho\times w'\right)$-continuous on a dense residual subset $\mathbb{E}' \subset \mathbb{E}$ if condition (SB) holds.

We conclude this section by noting additional applications of spaces of Bochner economies, specifically those involving \emph{continuous time horizons} and \emph{renewable resources over an infinite discrete time horizon}. The former utilizes the space $\mathscr L_\infty$, while the latter employs the sequence space $\ell_\infty$. Unlike the respective cases we shall explore in the next section, it is necessary to define commodity spaces as subspaces that are weakly compactly generated, i.e., they are the span of a weakly compact subset of $\mathscr L_\infty$ and $\ell_\infty$ respectively. For a more detailed explanation, refer to \cite{cea_continuity_2024}, p. 63, and for specific applications, see p. 67 in the cited reference.

\subsection{Gel'fand economies.}\label{GEc}

\begin{enumerate}

\item \textbf{Commodity differentiation. }\label{comdifferentiation}
    
Let $K$ be a compact metric space. We denote by $\mathcal{M}(K)$ the set of all finite signed Borel measures on $K$. Note that both $\mathcal{M}(K)_+$ and $C(K)_+$ are vector lattices. If $x \in \mathcal{M}(K)$, then $\|x\|$ denotes the variation norm of $x$, i.e., $\|x\| := |x|(K)$. When endowing $C(K)$ with the sup-norm topology, $\mathcal{M}(K)$ is isometrically isomorphic to the topological dual space of $C(K)$. We adopt the notation $\mathcal{M}(K) = C(K)^*$. At the same time, $C(K)$ is a subspace of $C(K)^{**} = \mathcal{M}(K)^*$. Thus, $\langle \mathcal{M}(K), C(K) \rangle$ represents the commodity-price pairing in this model.

Given the duality between $\mathcal{M}(K)$ and $C(K)$, the topology $w$ on $\mathcal{M}(K)$ will be the weak$^*$ topology $\sigma(\mathcal{M}(K), C(K))$. Thus, a sequence $(x_n)$ in $\mathcal{M}(K)$ $w$-converges to $x$ in $\mathcal{M}(K)$ if and only if $p \cdot x_n$ converges to $p \cdot x$ for all $p \in C(K)$, which is equivalent to saying that $\int_K p\cdot k dx_n$ converges to $\int_K p\cdot k dx$.

One can interpret each point in $K$ as representing a complete description of all relevant characteristics of one unit of a certain commodity. The commodity space is $\mathcal M(K)^w_+$ which is  separable and completely metrizable since $K$ is compact (\cite{varadajan_convergence_1958}, Theorem 3.4, p. 20). Let $\Omega \subset \mathcal{M}(K)_+$ be a $w = \sigma(\mathcal{M}(K), C(K))$-compact Polish subset. By adding additional assumptions,  Jones \cite{jones_existence_1983} proves the existence of distributional equilibria in $ S\times\mathcal M\left (\mathcal P'\times \Omega^w\times X\right )$ where $\mathcal P'$ is a compact subset of $\mathcal P_{\mathcal M(K)^+}$ and $\Omega\subset X$. The set $X$ is $w$-compact (\cite{jones_existence_1983} p. 135). In turn, $S = \left\{p\in C_+(K): \|p\|_\infty = 1\right\}$ is $w' = \sigma\left(C(K), \mathcal M(K)\right)$-compact. One of the assumptions in Jones \cite{jones_existence_1983} is that of \textit{uniform household smoothness} (UHS, for short).

\textbf{Assumption (UHS)}. For all $\theta > 1$, there is an $\xi > 0$ such that for all $\eta > 0$, for all $x\in X$, for all $\succ\in\mathcal P_{\mathcal M(K)^+}$, and for all $s, t\in K$ with $d_K(s, t) < \xi$, $$ x + \eta\theta\delta_t \succ x + \eta\delta_s$$

Where $d_K(\cdot, \cdot)$ is the metric on $K\times K$. We refer to \cite{jones_existence_1983} for the precise meaning of (UHS) on preferences. In a nutshell, this implies that in equilibrium, commodities with similar descriptions have similar prices. 

\begin{lemma}\label{uhs} Under Assumption \textnormal{(UHS)}, any sequence $(p_n)$ of price equilibria, $\|\cdot\|_\infty$-converges to $p$ in $S$. 
\end{lemma}

Therefore, (UHS) implies $\textnormal{lim}_n p_n\cdot e_n = p\cdot e$ for any sequence $(p_n)$ of equilibrium prices in $S$ and any sequence of endowments $(e_n)$ in $\Omega$ (see Remark \ref{jointcont}). 

Let $\mathcal P_X\subset \mathcal P'$ denote a compact set of preferences induced on $X$ satisfying (PR). From the preceding results and Proposition \ref{inducedequilibrium}, together with the subsequent discussion on the Borel representability of the induced individualized economy, we infer the existence of a distributional equilibrium $(\zeta, p) \in \mathcal M\left(\mathcal P_X \times \Omega^w \times X\right)\times S$, thus ensuring that the set of equilibria $\mathbb E \subset \mathcal M\left(\mathcal P_X \times \Omega^w\right)$ is nonempty. We are now ready to state the analogue of Proposition \ref{uhc} in this Gel'fand setting.

\begin{proposition}\label{uhcCD} The correspondence $DE:\mathbb E\twoheadrightarrow \mathcal M\left(T^w\times X\right)$ is upper hemi-continuous if Assumption \textnormal{(SB)} holds.
\end{proposition}

Jones shows that equilibrium prices are uniformly bounded below by some strictly positive parameter so condition (IPVE) is satisfied. The remaining results follow. 

\ 

\item\textbf{Allocations of renewable resources over an infinite discrete time horizon.}

When considering renewable resources, such as food or water, which have an upper bound on the amount available in each period, the appropriate set could be $\ell^+_\infty$. Let $X = \{x\in \ell^+_{\infty}\vert \ x_t\le\alpha \} $ be the common consumption set. Here, $x \in X$ represents an inter-temporal consumption bundle where each component $x_t$ is less than or equal to a very large assumed parameter $\alpha > 0$. Notice that with respect to the norm $\Vert x \Vert_{\infty} := \underset{t}{\text{sup}}\{\vert x_t\vert\}$, the set $X$ is bounded.

There are several topologies that one can consider on $\ell_\infty$. The weakest one is the product topology $\tau_d$ which is metrizable by $$d(x, x') = \sum_{t=1}^{\infty}\frac{\vert x_t - x'_t\vert}{2^t(1 + \vert x_t - x'_t\vert)}$$
for $x, x'\in \ell_\infty$. Since $\ell_\infty$ is the dual space of $\ell_1$, the topologies $w$ and $\tau$ are the weak$^*$-topology $\sigma(\ell_\infty, \ell_1)$ and the Mackey topology $\tau(\ell_\infty, \ell_1)$ on $\ell_\infty$ respectively. Furthermore, the set $X$ is $w$-compact by Alaoglu\'{}s Theorem. According to \cite{bewley_very_1991}, p. 226, the set $X$ is $\tau$-compact and metrizable by $\tau_d$. Moreover, $X^{w} = X^{\tau} = X^{\tau_d}$. Let $\Omega = X$ so $\mathcal P_X \times \Omega^\tau$ (resp. $\mathcal P_X \times \Omega^{w}$ and $\mathcal P_X \times \Omega^{\tau_d}$) is a compact Polish space for $\mathcal P_X$ compact. Hence, every $\mu$ in $\mathcal M(T^\tau)$ is a Gel'fand $\tau$-economy 

The topological dual of $\ell_\infty$ is $ba(\mathbb{N})$, the space of bounded and finitely additive set functions on $\mathbb N$ \begin{align*}
ba(\mathbb{N}) = \Big\{ \pi : 2^{\mathbb{N}} \to \mathbb{R} \ \Big| \
& \sup_{A \subset \mathbb{N}} |\pi(A)| < \infty, \\
& \pi(A \cup B) = \pi(A) + \pi(B) \quad \textnormal{whenever } A \cap B = \emptyset
\Big\}
\end{align*}

This space is a Banach space with the
norm $$\|\pi\| = \textnormal{sup}\left\{\sum_{i = 1}^n|\pi (A_i)| : A_i\cap A_j = \emptyset \ \textnormal{for} \ i\neq j\right\}$$

Let $\textbf{1} = (1, 1, ...)$. In general, prices belong to $\Delta := \{\pi\in ba^+(\mathbb{N}): \|\pi\| = \pi\cdot\textbf{1} = 1\}$, which is $\sigma (ba(\mathbb{N}), \ell_\infty)$-compact by Alaoglu's Theorem. Since $\ell_\infty$ is not separable, we cannot guarantee that the price set with the induced topology is metrizable.  

Since commodity vectors are modeled as sequences, it is mathematically natural to represent price vectors accordingly, rather than as set functions. The suitable space is $\ell_1$ which is a subspace of $ba(\mathbb{N})$ in the sense of being isometrically isomorphic to the space $ca(\mathbb{N})$ \begin{align*}
ca(\mathbb{N}) = \Big\{ 
\pi \in ba(\mathbb{N}) \ \Big| \ 
& \pi\left( \bigcup_{n \in \mathbb{N}} A_n \right) = \sum_{n \in \mathbb{N}} \pi(A_n) \\
& \textnormal{whenever } A_i \cap A_j = \emptyset \textnormal{ for } i \neq j
\Big\}
\end{align*}

Several results show the existence of subsets of $T^\tau$ such that every economy concentrated on it admits an equilibrium with prices in $\ell_1\cap \Delta$. For instance, the main theorems in \cite{suzuki_competitive_2013}, \cite{noguchi_economies_1997b} and \cite{bewley_very_1991}. Consequently, the set of Gel'fand economies $\mathbb{E} \subset \mathcal{M}(T^\tau)$ is nonempty.  Furthermore, the latter proves that equilibrium prices belong to a $w' = \sigma(\ell_1, \ell_\infty)$-compact subset $S$ of $\ell_1^+$.  

If preferences are strictly monotone, Assumptions (PR) and (SB) hold since $x = \alpha \mathbf{1}$ belongs to $X$. Furthermore, if $\beta \mathbf{1} = \inf \Omega$ for some $\beta > 0$\footnote{The idea is that $\beta$ is very close to zero.}, Assumption (IPVE) is also satisfied. Of course, condition (EP) clearly holds. Consequently, there exists a dense subset of $\mathbb{E}$ where the correspondence $DE:\mathbb E \twoheadrightarrow\mathcal M(T^\tau\times X)\times S$ is continuous.

\begin{example}
   Let $X = \left\{x \in \ell_\infty^+ \,\middle|\, x_t \le \alpha \text{ for all } t \right\}$ be the consumption set, and let $\Omega = \left\{x \in X \,\middle|\, x_t \ge \beta  \text{ for all } t \right\}$ be the set of initial endowments, where \( \alpha, \beta > 0 \). Let the preference relation \( \mathcal P_X = \{u_1, u_2, u_3 : X\rightarrow\mathbb R_+\}\) where :
\[
u_1(x) = \sum_{t=1}^\infty \delta^t x_t \quad \text{with } \delta \in (0,1),
\]
\[
u_2(x) = \sum_{t=1}^\infty 2^{-t} \log(1 + x_t),
\]
\[
u_3(x) = \sum_{t=1}^\infty \frac{1}{t^2} \cdot \frac{x_t^\gamma}{1 + x_t^\gamma} \quad \text{with } \gamma \in (0,1).\]

Each of these utility functions is increasing, concave and continuous on $X$ with respect to the topology $w' = \sigma(\ell_\infty, \ell_1)$.

The set of characteristics $\mathcal P_X \times \Omega$ is compact under the product of the closed convergence topology on $\mathcal P_X$ and the Mackey topology on $\Omega$. Hence, condition \textnormal{(PC)} is satisfied. Since prices lie in $\ell_1$, every distributional economy $\mu \in \mathcal{M}\left(\mathcal P_X \times \Omega\right)$ is Gel'fand integrable. Furthermore, because the endowment set admits a strictly positive minimum, condition \textnormal{(IPVE)} holds. The strict monotonicity of preferences implies that the consumption set has a maximum; thus, each $\mu$ satisfies both \textnormal{(PR)} and \textnormal{(SB)}. In this setting, $\Omega$ is endowed with the Mackey topology $\tau$, and the simplex $S$ is $w'$-compact. Actually, $S$ is $\|\cdot\|_1$-compact since $\ell_1$ has the Schur property. Consequently, the mapping $(e, p) \mapsto p \cdot e$ is continuous on $\Omega \times S$, ensuring that condition \textnormal{(EP)} is fulfilled.
\end{example}

\item \textbf{Economies with continuous time horizon.}\label{UCT}

We now consider a continuous, unbounded time horizon approach by examining the space $\mathscr L_{\infty}(M, \mathscr M, \nu)$ of (essentially) real-valued, $\mathscr{M}$-measurable, and $\nu$-essentially bounded functions on $M$. For brevity, we denote this space simply as $\mathscr L_{\infty}$. In addition to the norm topology $\|\cdot\|_\infty$, we endow $\mathscr L_{\infty}$ with the weak$^*$ topology $w$, and the Mackey topology $\tau$, consistent with the duality $\langle \mathscr L_{\infty}, \mathscr L_1 \rangle$. The space $\mathscr L_{\infty}^\tau$ which is the space $\mathscr L_{\infty}$ endowed with the topology $\tau$ is a Lusin space; in other words, it is a regular topological space that admits a stronger topology under which it becomes Polish. Naturally, $\mathscr L_{\infty}^w$ is also Lusin. When we consider the norm $\|f\|_\infty = \text{inf}\{r \mid \nu\{m \mid f(m) < r\} = 0 \}$, $\mathscr L_{\infty}$ becomes a Banach space whose topological dual is $ba(M, \mathscr M, \nu)$, the space of bounded additive set functions on
$\mathscr M$ that vanish on sets of $\nu$-measure zero. For every $x\in\mathscr L_\infty$ and $\pi \in ba(M, \mathscr M)$, $\pi\cdot x = \int_M xd\pi$. By the Radon-Nikodyn Theorem, $\mathscr L_1$ can be identified with the subspace of $ ba(M, \mathscr M, \nu)$ whose elements are countably additive. Hence, by considering $\mathscr L_1\subset ba(M, \mathscr M, \nu) = \mathscr L_\infty^* = \mathscr L_1^{**}$ we can endow $\mathscr L_1$ with the weak topology $w' = \sigma(\mathscr{L}_1, \mathscr{L}_\infty)$.

Let $X$ be a norm-bounded, convex subset of $\mathscr{L}_\infty^+$ containing the null vector. By Alaoglu's Theorem, $X$ is weakly$^*$-compact and since $\mathscr{L}_1$ is separable, $X$ is metrizable so it is $w$-compact Polish. Moreover, since $X$ is bounded in the Mackey topology $\tau$, it follows that $X$ is $\tau$-Polish (\cite{mertens_equivalence_1991}, p.~191). The convexity of $X$ implies that it is $\tau$-closed. In conclusion, $X$ is both a $w$-compact Polish space and a $\tau$-closed, $\tau$-Polish space. Let $\Omega$ be a $\tau$-compact subset of $X$ of initial endowments. Clearly, $\Omega$ is $\tau$-Polish. Since both $\mathcal{P}_X$ and $\Omega^\tau$ are compact Polish spaces, the product $T^\tau = \mathcal{P}_X \times \Omega^\tau$ is also compact and Polish.
 
 Let $S \subset \mathscr L^+_1$ be a $w'$-compact space that represents the price simplex. This space is metrizable since $\mathscr L_1$ is separable. Every Borel probability distribution $\mu$ on $T^\tau$ can be regarded as a Gel'fand $\tau$-economy that satisfies Assumption (EP). 

According to Theorem 6 in \cite{zame_markets_1986}, under appropriate assumptions, individualized economies admit equilibria with allocations in order intervals (see the proof on pp.~10\textendash22 to 10\textendash26). Without loss of generality, we may assume that this interval lies within $X$. By Proposition \ref{inducedequilibrium} and Lemma~\ref{wequivalence}(1), we then deduce the existence of a pair $(\zeta, p) \in \mathcal{M}(\mathcal P_X\times\Omega^\tau \times X) \times S$ that constitutes a competitive equilibrium for $\mu$, thereby proving the non-emptiness of $\mathbb{E}$.

Under the assumptions of Theorem~\ref{contDW}, it further follows that there exists a dense and residual subset of $\mathbb{E}$ consisting of equilibria that are $(\rho_T, \rho \times w')$-continuous.

\end{enumerate}
\appendix
\begin{appendices}

\renewcommand{\thetheorem}{\Alph{section}.\arabic{theorem}}
\renewcommand{\thelemma}{\Alph{section}.\arabic{lemma}}
\renewcommand{\theproposition}{\Alph{section}.\arabic{proposition}}
\renewcommand{\thecorollary}{\Alph{section}.\arabic{corollary}}
\renewcommand{\thedefinition}{\Alph{section}.\arabic{definition}}
\renewcommand{\theremark}{\Alph{section}.\arabic{remark}}
\renewcommand{\theexample}{\Alph{section}.\arabic{example}}
\numberwithin{equation}{section}

\section{Proofs.}

\subsection*{Lemma \ref{contpref}.}

\ 

Before proving this lemma, we recall some useful characterizations.
Let $(F_n)$ be a sequence of subsets of a metric space $M$.
A vector $\zeta$ belongs to the subset $\textnormal{Li}(F_n)$ of $M$ if and only if there exist a sequence $(\zeta^n)$ and an integer $n_0$ such that $\zeta^n \in F_n$ for all $n \ge n_0$ and $\text{lim}_n \zeta^n = \zeta$.
Similarly, $\zeta$ belongs to the subset $\textnormal{Ls}(F_n)$ of $M$ if and only if there exists a subsequence $(F_{n_m})$ and, for each $m$, an element $\zeta_{n_m} \in F_{n_m}$ such that $\lim_m \zeta_{n_m} = \zeta$.
Both $\textnormal{Li}(F_n)$ and $\textnormal{Ls}(F_n)$ are closed, and $\textnormal{Li}(F_n) \subset \textnormal{Ls}(F_n)$.
We say that a set $F \subset M$ is the \textit{closed limit} of $(F_n)$ when $F = \textnormal{Li}(F_n) = \textnormal{Ls}(F_n)$.
For details, see \cite{hildenbrand_core_1974}, p.~15.
\
Returning to the proof of the lemma, consider the set 

\[
A = \left\{(\succ, x, y)\in \mathcal{P} \times X \times X : (x,y) \in h(\succ)\right\}
\]
We first show that $A$ is closed in the product topology.
Let $(\succ^n, x^n, y^n) \to (\succ, x, y)$ in $\mathcal{P} \times X \times X$, with $(x^n, y^n) \in h(\succ^n)$ for every $n$.
Since $X$ is weakly closed, it follows that $x, y \in X$. For the same reason, $\succ\in \mathcal P$.
Hence, it remains to prove that $(x, y) \in h(\succ)$.
We remark that $h(\succ^n)$ converges to $h(\succ)$ in $\mathcal{C}_w(X \times X)$ if and only if $(\succ^n)$ converges to $(\succ)$ in $\mathcal{P}$.
Furthermore, since $\mathcal{C}_w(X \times X)$ is a compact metric space, convergence of graphs $h(\succ^n)$ to $h(\succ)$ is characterized by the closed limit as described above (see \cite{hildenbrand_core_1974}, Theorem~2, p.~19).
Therefore, it suffices to verify that $(x, y) \in \textnormal{Ls}\left(h(\succ^n)\right)$, which is immediate.
Thus, $(x, y) \in h(\succ)$, showing that $A$ is closed in $\mathcal{P} \times X \times X$.

Finally, since the complement $\mathcal{P} \times X \times X \setminus A$ is open, any point $(\succ, x, y)$ such that $x \succ y$ clearly belongs to it.

\subsection*{Proposition \ref{uhc}.}

\ 

Before starting with the proof, we state the following lemma:

\begin{lemma}\label{indequi} Let $(A, \mathcal A, \nu)$ be a $\sigma$-finite measure space. Let the sequence $(\mathcal E^n, x^n): (A, \mathcal A)\rightarrow \left(T\times X, \ \mathcal B(T^\kappa)\otimes\mathcal B(X)\right) $ and $(\mathcal E, x): (A, \mathcal A)\rightarrow \left(T\times X, \ \mathcal B(T^\kappa)\otimes\mathcal B(X)\right)$ be measurable functions such that there exists a sequence $(p^n)\in S$ where $(x_n, p_n)$ is an individualized equilibrium for $\mathcal E^n$ for every $n$. Suppose that $ \textnormal{lim}_n (\mathcal E_n(a), x_n(a)) = (\mathcal E(a), x(a))$ a.e. for the product topology $\mathcal T_{X}^C\times \kappa\times w$ and that $\textnormal{lim}_n p_n = p$ in $\kappa'$, then $(x, p)$ is an individualized equilibrium for $\mathcal E$ under the assumptions of Proposition \ref{uhc}. \end{lemma}

\begin{proof} We recall that under Assumption (EP) and (IPVE), we have that $\textnormal{lim}_np_n\cdot e_n(a)\ge p\cdot e (a) > 0$. a.e. Furthermore, $\int_A x_n(a)d\nu = \int_A e_n(a)d\nu$ for every $n$ implies $\int_A x(a)d\nu = \int_A e(a)d\nu$ by Corollary \ref{cKDCT} (resp. by Proposition \ref{GDCT} if $L$ is a dual space). It only remains to show that $(x, p)$ is a Walrasian allocation for $\mathcal E$. We shall prove this in a serie of steps. For simplicity, most of the time we will omit the subscript $a$. 

\ 

\noindent {\sc Step 1.} $D(p, \succ, e)$ \emph{is a non-empty $w$-compact-valued correspondence}. The nonemptiness follows from Lemma 2 in \cite{lee_competitive_2013}. As for compactness, it suffices to show that $D(p, \succ, e)$ is closed in $X$. Let $(\xi_n)\in D(p, \succ, e)$ be a sequence that $w$-converges to $\xi$. If $\xi\not\in D(p, \succ, e)$, there exists $x'\in X$ such that $x'\succ \xi$ and $p\cdot x' < p\cdot e$ without loss of generality. For $n$ large enough, $x'\succ \xi_n$ due to Lemma \ref{contpref} but this contradicts that $\xi_n\in D(p, \succ, e)$. 

\ 

From Remark \ref{jointcont}, we understand that, in general, $p_n \cdot x_n$ does not converge to $p \cdot x$. To address this difficulty, we follow the technique employed in \cite{podczeck_markets_1997}, among others, which consists in ``enlarging the demand set.

\ 

For a given $(p, \succ, e)$ let us consider the \emph{better-than-set} $$\overset{\sim}{D}(p, \succ, e) = X\setminus\left\{x\in X: \ \exists \ x'\in D(p, \succ, e) \ \text{such that} \  x'\succ x\right\}$$

It is obvious that $D(p, \succ, e)\subset \overset{\sim}{D}(p, \succ, e)$. Since preferences are continuous, $\overset{\sim}{D}(p, \succ, e)$ is $w$-closed, whence $w$-compact.

\ 

\noindent {\sc Step 2.} \emph{$x\in\overset{\sim }{D}(p, \succ, e)$}. Suppose not, then there exists $\zeta\in X$ such that $\zeta\succ x$ and $p\cdot \zeta\le p\cdot e$. By continuity of preferences we can take $p\cdot \zeta < p\cdot e$ without loss of generality. Consequently, there exists $n_0$ such that for every $n\ge n_0$, $\zeta\succ_n x$ by Lemma \ref{contpref}. Since $x_n\in D(p_n, \succ_n, e_n)$, we have $p^n\cdot \zeta > p^n\cdot e^n$ so by Assumption (EP), $p\cdot \zeta \ge p\cdot e$,  a contradiction.


\

\noindent {\sc Step 3.} $p\cdot x\ge p\cdot e$. If $x$ is a satiation point, $p\cdot x\ge p\cdot e$ by Assumption (SB). If $x$ is not a satiation point, there exists $x'$ as close as one wants to $x$ such that $x'\succ x$ by Assumption PR (2). By continuity of preferences (Lemma \ref{contpref}), $x'\succ_n x_n$ for every $n > n_0$. It follows that $p_n\cdot x' > p_n\cdot e_n$ because of equilibrium conditions in $\mathcal E^n$. Taking limits we get $p\cdot x'\ge p\cdot e$ by condition (EP). Since $x'$ is arbitrarily close to $x$, the Step is proved. 

\ 

We close the proof with the following step.

\ 

\noindent {\sc Step 4.} \emph{$x(a) \in D(p, \succ_a, e(a))$ a.e.} We note that the demand correspondence is $D(p, \succ_a, e(a)) = \overset{\sim}{D}(p, \succ_a, e(a))\cap \left\{x\in L_+ : p\cdot x(a) = p\cdot e(a)\right\}$ by PR (2). Suppose that there exists $A'\subset A$ with $\nu (A') > 0$ and $p\cdot x(a) > p\cdot e(a)$ for every $a\in A'$. But this implies $p\cdot\int_A x(a) d\nu > p\cdot\int_A e(a) d\nu$ which contradicts the fact that $\int_A x(a)d\nu = \int_A e(a)d\nu$.
\end{proof}

\

We now come to the proof Proposition \ref{uhc}. First, since the marginal of $\zeta_n$ with respect to the space of characteristics, $\zeta_{n|T}$, is equal to $\mu_n$ for all $n$ by definition, we have that $\zeta_{T} = \mu$. We proceed to show that $\zeta\in DE(\mu)$. Indeed, for each $n\in \mathbb N$, there exists $p_n\in S$ such that $\zeta_n(E_{p_n}) = 1$. Since $S$ is $\kappa'$-compact, there exists a subnet also denoted $(p_n)$ that converges to $p$. By Skorokhod's theorem, (\cite{billingsley_convergence_1999}, Theorem 6.7, p. 70, see also, \cite{hildenbrand_core_1974}, p. 50) for the sequence $(\zeta_n)$, one can take the Lebesgue unit interval space $(I, \mathcal{B}(I), \lambda)$ and measurable mappings $(\mathcal{E}_n, \ x_n)$ and $(\mathcal{E}, \ x)$ from $I$ into $T^\kappa\times X$ such that $(\mathcal{E}_n(i), \ x_n(i))$ $\mathcal T^C_X\times\kappa\times w$-converges to $(\mathcal{E}(i), \ x(i))$ a.e., and for all $n$, $\lambda \circ  (\mathcal{E}_n, x_n)^{-1} = \zeta_n$ and $\lambda\circ\mathcal E_n
^{-1} = \mu_n$. Furthermore, $\lambda \circ (\mathcal E, x)
^{-1} = \zeta$ and $\lambda \circ \mathcal E^{-1} = \mu$. Hence, $\lambda \circ  (\mathcal{E}_n, x_n)^{-1} \in DE(\lambda\circ\mathcal E_n^{-1})$ for all $n$. Note that $(\mathcal{E}_n, \ x_n)$ and $(\mathcal{E}, \ x)$ represent a sequence and a limit point of individualized economies whose distributions coincide with $(\zeta_n)$ and $\zeta$ respectively. Hence, by Corollary \ref{cBintegration} (resp. Lemma \ref{Gintegration}), the mappings $e_n :=\text{Proj}_{\Omega}\mathcal E_n$, $e :=\text{Proj}_{\Omega}\mathcal E$, $x_n$ and $x$ are Bochner integrable (resp. Gel'fand integrable). Since $\int_X\iota d\zeta_{n|X} = \int_\Omega\iota d\mu_{n|\Omega}$ for all $n$, it follows that $\int_I x_n(i) d\lambda = \int_I e_n(i) d\lambda$ for all $n$ by Corollary \ref{chvariable} (resp. by Corollary \ref{wchvariable}). Furthermore, by Corollary \ref{cKDCT} (resp. Proposition \ref{GDCT}), we have $\int_I x(i) d\lambda = \int_I e(i) d\lambda$. By Lemma \ref{wequivalence} (2), $(x_n, p_n)$ is a Walrasian equilibrium for $\mathcal E_n$ whence $(x, p)$ is an individualized equilibrium for $\mathcal E$ by Lemma \ref{indequi}. In virtue of Corollary \ref{chvariable} (resp. by Corollary \ref{wchvariable}),  $\int_X\iota d\zeta_{X} = \int_\Omega\iota d\mu_{\Omega}$ and due to Lemma \ref{wequivalence} (2), $\lambda\circ (\mathcal E, x)^{-1} = \zeta$ is a distributional equilibrium for the distributional economy $\lambda\circ \mathcal E^{-1} = \mu$. 

\begin{remark}

  We emphasize that Lemma \ref{indequi} is not only instrumental for the proof of Proposition \ref{uhc}, but it is also a result of independent interest. Specifically, it establishes that the equilibrium correspondence is closed for individualized economies defined over a common agent space. This observation is significant because, in the literature, the existence of equilibria is sometimes established by demonstrating that the equilibrium correspondence is closed for particular sequences of economies and associated equilibrium vectors. This approach appears, for example, in \cite{khan_equilibria_1991}, pp. 244-245, among others. In that context, one considers a fixed endowment function $e_n = e: A \to \Omega$ for all $n$, so that the map $(e, p) \mapsto p \cdot e$ is jointly continuous when the product topology $w \times w'$ is imposed on $S \times \Omega$. This suggests that a stronger assumption than (EP) could be to restrict attention to families of economies sharing the same marginal distribution $\mu_\Omega$. 
\end{remark}

\subsection*{Corollary \ref{baire}.}

\ 

Since $\mathcal M(T^\kappa)$ is completely metrizable, it suffices to show that $\mathbb E$ is closed to invoke the Baire category theorem. If $\mathbb E = \emptyset$, we do not have to prove anything. Let $(\mu_n)$ be a sequence in $\mathbb E$ which converges to $\mu$. Hence, there exists a sequence $(\zeta_n)$ in $\mathcal M(T^\kappa\times X)$ such that $\zeta_n\in DE(\mu_n)$ for all $n\in\mathbb N$. Since $\mathcal M(T^\kappa\times X)$ is compact there exists a subsequence of $(\zeta_n)$ (say, itself) which converges to $\zeta$. By Theorem \ref{uhc}, $\zeta$ belongs to $DE(\mu)$ and thus $\mu$ belongs to $\mathbb E$.

\subsection*{Proposition \ref{wclosenessindiv}.}

\ 

Notice that $\nu_n\circ (\mathcal E_n, x_n)^{-1}$ belongs to $DE\left(\nu_n\circ\mathcal E_n^{-1}\right)$ for each $n\in\mathbb N$ accordingly to Lemma \ref{wequivalence} (1). Furthermore, $(\nu_n\circ \mathcal E_n^{-1})$ converges to $\nu\circ\mathcal E^{-1}$ and $(\nu_n\circ x_n^{-1})$ converges to $\nu\circ x^{-1}$. By Theorem 2.1 (iii) in \cite{keisler_why_2009}, there exists a subsequence of $(\nu_n\circ (\mathcal E_n, x_n)^{-1})$ (for convenience, itself) which converges to $\nu\circ(\mathcal E, x)^{-1}$. By Proposition \ref{uhc}, $\nu\circ(\mathcal E, x)^{-1}$ belongs to $DE\left(\nu\circ\mathcal E^{-1}\right)$. Hence, $x$ is an equilibrium for $\mathcal E$ by Lemma \ref{wequivalence} (2). 

\subsection*{Proposition \ref{closenessindiv}.}

\ 

As in the previous case, $\nu_n\circ (\mathcal E_n, x_n)^{-1}$ belongs to $DE\left(\nu_n\circ\mathcal E_n^{-1}\right)$ for each $n\in\mathbb N$. The sequence $(\nu_n\circ \mathcal E_n^{-1})$ converges to $\nu\circ\mathcal E^{-1}$ and $(\nu_n\circ x_n^{-1})$ converges to the measure $\gamma\in\mathcal M(X)$. By Theorem 2.1 (iii) in \cite{keisler_why_2009}, there exists a subsequence of $\nu_n\circ (\mathcal E_n, x_n)^{-1}$ (say, itself) which converges to $\zeta\in\mathcal M(T^\kappa\times X)$ such that the marginals $\zeta_{T}$ and $\zeta_{X}$ are $\nu\circ\mathcal E^{-1}$ and $\gamma$, respectively. By Proposition \ref{uhc}, $\zeta$ belongs to $DE\left(\nu\circ\mathcal E^{-1}\right)$ and since $(A, \mathcal A, \nu)$ is saturated, there exists $x: A\rightarrow X$ such that $\nu\circ (\mathcal E, x)^{-1} = \zeta$ and $\zeta_{X} =  \nu\circ x^{-1} = \gamma$. Hence, $x$ is an equilibrium for $\mathcal E$ by Lemma \ref{wequivalence} (2). 

\subsection*{Lemma \ref{inducedpreference}.}
Define $R = h^{-1}\circ r\circ g$ where $r:\mathcal C_w(L_+\times L_+)\rightarrow \mathcal C_w(X\times X)$ is given by $r(F) = F\cap (X\times X)$. Notice that, due to the topologies on $\mathcal P_{L_+}$ and $\mathcal P_X$, the map $g$ is continuous and $h:\mathcal P_X\rightarrow h(\mathcal P_X)$ is a homeomorphism. Consequently, both $g$ and $h^{-1}$ are Borel measurable. Thus, to complete the proof, it suffices to show that $r$ is Borel measurable.

\begin{claim} The function $r: \mathcal C_w(L_+\times L_+)\rightarrow \mathcal C_w(X\times X)$ is Borel measurable.
\end{claim}
\begin{proof}

Let
\[
\mathscr S_{X\times X}
=
\bigl\{[K]_{X\times X}: K\subset X\times X \text{ compact}\bigr\}
\cup
\bigl\{\langle G\rangle_{X\times X}: G\subset X\times X \text{ open}\bigr\}
\]
be a subbase for the closed convergence topology on $\mathcal C_w(X\times X)$. Since the Borel $\sigma$-algebra is generated by the topology, and the topology itself is generated by $\mathscr S_{X\times X}$, we have $\mathcal B(\mathcal C_w(X\times X))=\sigma(\mathscr S_{X\times X})$. Hence, it is enough to prove that the preimage under $r$ of every member of $\mathscr S_{X\times X}$ is a Borel subset of $\mathcal C_w(L_+\times L_+)$.

\medskip

First, let $K\subset X\times X$ be compact. Then
\begin{align*}
r^{-1}([K]_{X\times X}) 
&= \bigl\{F\in\mathcal C_w(L_+\times L_+) : (F\cap (X\times X))\cap K=\varnothing\bigr\} \\
&= \bigl\{F\in\mathcal C_w(L_+\times L_+) : F\cap K=\varnothing\bigr\},
\end{align*}
where the last equality holds because $K\subset X\times X$. Since $K$ is compact in $X\times X$ and $X\times X$ is weakly compact, $K$ is also a compact subset of $L_+\times L_+$ with respect to the product weak topology. Therefore $r^{-1}([K]_{X\times X}) = [K]_{L_+\times L_+}$, which is a subbasic open set in $\mathcal C_w(L_+\times L_+)$ and hence Borel.

Now, let $G\subset X\times X$ be open. Because $G\subset X\times X$, for every $F\in\mathcal C_w(L_+\times L_+)$ we have $(F\cap (X\times X))\cap G = F\cap G$. Thus
\[
r^{-1}(\langle G\rangle_{X\times X}) = \{F\in\mathcal C_w(L_+\times L_+):F\cap G\neq\varnothing\}. \tag{1}
\]

Fix a metric $d$ on $X\times X$ that generates its weak topology (for instance, the product metric). Set $H = (X\times X)\setminus G$, which is closed in $X\times X$. For each $n\ge 1$, define
\[
K_n:=\{x\in X\times X:d(x,H)\ge 1/n\}.
\]

We now establish two facts.

\medskip
\noindent\textbf{(i)} Each $K_n$ is compact.  
Indeed, the function $f(x)=d(x,H)$ is continuous because it is $1$-Lipschitz: for any $x,y\in X\times X$,
\[
|f(x)-f(y)| = |d(x,H)-d(y,H)| \le d(x,y).
\]
Moreover, $K_n = f^{-1}([1/n,\infty))$ is the preimage of a closed set under a continuous map, hence closed in $X\times X$. Since $X\times X$ is weakly compact, $K_n$ is also weakly compact.

\medskip
\noindent\textbf{(ii)} $G = \bigcup_{n=1}^{\infty} K_n$.  
If $x\in K_n$ for some $n$, then $d(x,H)\ge 1/n >0$, so $x\notin H$ and therefore $x\in (X\times X)\setminus H = G$. Conversely, if $x\in G$, then $x\notin H$. Because $H$ is closed, $d(x,H)>0$. Choose $n$ such that $1/n \le d(x,H)$; then $x\in K_n$.

\medskip
Using (1) and (ii), we obtain
\begin{align*}
r^{-1}(\langle G\rangle_{X\times X}) 
&= \left\{F\in\mathcal C_w(L_+\times L_+):F\cap \left(\bigcup_{n=1}^\infty K_n\right)\neq\varnothing\right\} \\
&= \bigcup_{n=1}^\infty \left\{F\in\mathcal C_w(L_+\times L_+):F\cap K_n\neq\varnothing\right\} \\
&= \bigcup_{n=1}^\infty \left(\mathcal C_w(L_+\times L_+)\setminus [K_n]_{L_+\times L_+}\right).
\end{align*}

Actually, each $K_n$ is compact in $L_+\times L_+$, so $[K_n]_{L_+\times L_+}$ is a (subbasic) open set in the closed convergence topology on $\mathcal C_w(L_+\times L_+)$. Consequently, each $\mathcal C_w(L_+\times L_+)\setminus [K_n]_{L_+\times L_+}$ is closed. Hence $r^{-1}(\langle G\rangle_{X\times X})$ is a countable union of closed sets and therefore Borel.

\end{proof}

\subsection*{Proposition \ref{inducedequilibrium}.}

\ 

We first show that such an economy $\mathcal E_X$ exists. In fact, define $\mathcal E_X := \bar R \circ \mathcal E$, where $\bar R = (R, \iota): \mathcal P_{L_+} \times \Omega \to \mathcal P_X \times \Omega$ by $\bar R(\succ, e) = (R(\succ), \iota(e)) = (\succ_X, e)$ for all $(\succ, e)\in \mathcal P_{L_+}\times\Omega$. We note that $\mathcal E_X$ is measurable since both $R$ and $\mathcal E$ are measurable by Claim \ref{inducedpreference}. Furthermore, it follows that $\textnormal{Proj}_\Omega\mathcal E_X = \textnormal{Proj}_\Omega\mathcal E$ so $\left\|\int_A\textnormal{Proj}_\Omega\mathcal E_X(a)d\nu\right\| < \infty.$

Second, since $(x, p)$ is an equilibrium allocation for $\mathcal E$, we trivially have for $\mathcal E_X$ the following properties: (i) $\int_A x(a)d\nu(a) = \int_A e(a) d\nu(a)$ and (ii) $p \cdot x(a) = p \cdot e(a)$ and $p \cdot v > p \cdot e(a)$ for all $v \in X$ such that $v \succ_X x(a)$ a.e. 
    
\subsection*{Lemma \ref{uhs}.}

\ 

It suffices to show that $(p_n)$ is equicontinuous. This is so because the Ascoli-Arzel\'a Theorem (\cite{royden_real_1988}, Theorem 40, p. 169. See also Corollary 41.) implies that the sequence $(p_n)$ converges uniformly to $p\in C(K)$, i.e., for any $\varepsilon > 0$ there exists $n_0$ such that for all $n \ge n_0$, $\|p_n - p\|_\infty\le\varepsilon$.

We recall that $(p_n)$ is equicontinuous if for every $\varepsilon > 0$ there exists $\delta > 0$ such that for all $n$ and all $k, k' \in K$ with $ d_K(k, k') < \delta$, $|p^n\cdot k - p^n\cdot k'| < \varepsilon$. In term of sequences, for every $(k^m)$ and $(k'^m)$ converging to $k$ such that there is a subsequence $(p^{n_m})$ and $p^{n_m}\cdot k^m - p^{n_m}\cdot k'^m $ converges to $0$ as $m$ tends to infinity (\cite{jones_existence_1983}, Lemma 2) given the fact that the sequence $(p_n)$ is uniformly bounded. 

Suppose, on the contrary, that $\text{lim}_m\frac{p^{n_m}\cdot k^m}{p^{n_m}\cdot k'^m} > 1$.  Let us consider the subsequence $(x_{n_m})$ and let $B_k$ be an open measurable set containing $k$ such that $\left(\int_I e(i)d\lambda(i)\right)(B_k) > 0$. Since $\int_I e_{n_m}(i)d\lambda(i)$ $w$-converges to $\int_I e(i)d\lambda(i)$ (Proposition \ref{GDCT}), there are infinitely many $m$ such that $\left(\int_I e_{n_m}(i)d\lambda(i)\right)(B_k) > 0$. This implies that $\left\{i\in I: x_{n_m}(i)(B_k) > 0\right\}$ has strictly positive measure for all $m$. With this in mind, let $$ x'(i)^{m} = x^{n_m}(i) - x^{n_m}(i)(B_k)\delta_{k^m} + \left ( \frac{p^{n_m}\cdot k^m}{p^{n_m}\cdot k'^m}\right )x^{n_m}(i)(B_k)\delta_{k'^m}  $$ where $\delta_{k^m}$ and $\delta_{k^m}$ are the respective Dirac deltas of $k^m$ and $k'^m$. Of course, $x'(i)^m$ is measurable and $p_{n_m}\cdot x'_m(i) = p_{n_m}\cdot x_{n_m}(i)$ a.e. By Assumption (UHS) one deduces $x'_m(i) \succ_i^m x_{n_m}(i)$ a.e. for $m$ large enough. Since $x^{n_m}(i)(B_k) > 0$ a.e., we have a contradiction with the fact that $x_{n_m}(i)\in D(\succ_i^{n_m}, e(i)_{n_m}, p_{n_m}) $ a.e.

\subsection*{Proposition \ref{uhcCD}.}

\ 

Let $(\mu_n)$ be a sequence of economies in $\mathbb E$ converging to $\mu$. Let $(\zeta^n)$ be a sequence of equilibria distribution on $T^w\times X$. There exists a subsequence also denoted by $(\zeta^n)$ which converges to $\zeta \in \mathcal M\left (T^w\times X\right )$. We have to show that $\zeta\in DE(\mu)$.

By repeating the proof of Proposition \ref{uhc}, we know that there exists a sequence $(x_n, e_n, p_n)$ converging to $(x, e, p)$ where $x_n, x : I\rightarrow X$, $e_n, e: I\rightarrow\Omega^w$ and $(x_n, p_n)$ is an equilibrium for $\mathcal E_n$ for all $n$. Given the assumptions of the proposition together with the fact that the condition (IPVE) is satisfied, if we show that $\textnormal{lim}_n p_n\cdot e_n(i) = p\cdot e(i)$ and $\textnormal{lim}_n p_n\cdot x_n(i) = p\cdot x(i)$ a.e., then we are done (see the proof of Proposition \ref{uhc}). So we end by proving the following claim.

\ 

\noindent {\sc Claim} \textit{$\textnormal{lim}_n p_n\cdot e_n(i) = p\cdot e(i)$ and $\textnormal{lim}_n p_n\cdot x_n(i) = p\cdot x(i)$} a.e.

\ 

Since for all $i$, $e_n(i)$ $w$-converges to $e(i)$, $p\cdot e_n(i)\rightarrow p\cdot e(i)$ and $e_n(i)(K)\le e(i)(K) + \varepsilon$, (\cite{hildenbrand_core_1974}, (26), p. 48). Hence, 

\ 

$|p_n\cdot e_n(i) - p\cdot e(i)| \le |p_n\cdot e_n(i) -p\cdot e_n(i)| + |p\cdot e_n(i) - p\cdot e(i)|\le\|p_n - p\|_\infty\|e_n(i)\| +
|p\cdot e_n(i) - p\cdot e(i)| \le\varepsilon (e(i)(K) + \varepsilon) +
|p\cdot e_n(i) - p\cdot e(i)|$.

\ 

So, $\text{lim}_n p_n\cdot e_n(i) = p\cdot e(i)$ a.e. In the same way, $\text{lim}_n p_n\cdot x_n(i) = p\cdot x(i)$ a.e.

\section{Some results on Bochner and Gel'fand integrals.}\label{measurability}

\ 

In this section, we present the mathematical statements employed in this paper. We emphasize integration as it is fundamental to defining feasibility in both distributional and individualized economies. We offer proofs here for convenience but claim no originality. We start with the following result, which is crucial for showing separability in the strong topology. 

\begin{lemma}\label{separability}
Let $E$ be a Hausdorff locally convex space and let $\mathcal T\subset \mathcal T'$ be two compatible locally convex topologies on $E$. Let $Y\subset E$ be convex. Then $Y$ is separable with respect to $\mathcal T$ if and
only if it is separable with respect to $\mathcal T'$.
\end{lemma}

\begin{proof}
If $Y$ is $\mathcal T'$-separable, there exists a countable subset $D'$ of $Y$ such that $\overline{D'}^{\mathcal T'} = Y$.  The $\mathcal T$-separability comes from the fact that $\overline{D'}^{\mathcal T'} \subset \overline{D'}^{\mathcal T}$. Conversely, let $Y$ be $\mathcal T$-separable. There exists a countable $\mathcal T$-dense subset $D$ of $Y$. Let $\text{co}(D)$ be the convex hull of $D$ and $\text{co}_{\mathbb Q}(D)$ be the rational convex hull of $D$, i.e., $\text{co}_{\mathbb Q}(D)$ contains only rational convex combinations of $D$. Obviously, $\text{co}_{\mathbb Q}(D)\subset \text{co}(D)$.

We claim that for any topology, $\mathcal T$ or $\mathcal T'$, the closure of $\text{co}_{\mathbb Q}(D)$ coincides with the closure of $\text{co}(D)$. For it, it suffices to show that $\text{co}(D)$ is contained in the closure of $\text{co}_{\mathbb Q}(D)$. Indeed, Suppose \( y \in \text{co}(D) \). Then there exist \( x_1, \dots, x_n \in D \) and real coefficients \( \lambda_1, \dots, \lambda_n \in [0,1] \) such that \( \sum_{i=1}^n \lambda_i = 1 \) and \( y = \sum_{i=1}^n \lambda_i x_i \). Since \( \mathbb{Q} \cap [0,1] \) is dense in \( [0,1] \), there exist sequences \( (\lambda_i^k) \in \mathbb{Q} \cap [0,1] \) such that \( \sum_{i=1}^n \lambda_i^k = 1 \) and \( \lambda_i^k \to \lambda_i \) for all \( i = 1, \dots, n \)\footnote{For each \( i = 1, \dots, n \), choose a sequence of rational numbers \(( \tilde{\lambda}_i^k) \in \mathbb{Q} \cap [0,1] \) such that \(\text{lim}_k \tilde{\lambda}_i^k \to \lambda_i \). These sequences exist since $\mathbb Q\cap[0, 1]$ is dense in $[0, 1]$. Define the normalization factor \( S_k := \sum_{i=1}^n \tilde{\lambda}_i^k \), so that \( S_k \to 1 \). 
Now set \( \lambda_i^k := \tilde{\lambda}_i^k / S_k \in \mathbb{Q} \cap [0,1] \). These satisfy \( \sum_{i=1}^n \lambda_i^k = 1 \), and \(\text{lim}_k \lambda_i^k \to \lambda_i \) for all \( i \).}. Define \( y_k := \sum_{i=1}^n \lambda_i^k x_i \in \text{co}_{\mathbb{Q}}(D) \). Hence, \( y_k \to y \) which implies that $y$ belongs to the closure of $\text{co}_{\mathbb{Q}}(D) \). Since $D\subset \text{co}(D)\subset Y = \overline{D}^{\mathcal T}$, it follows that $ \overline{\text{co}(D)}^{\mathcal T} = Y$. Due to the convexity of $\text{co}(D)$, $\overline{\text{co}(D)}^{\mathcal T} = \overline{\text{co}(D)}^{\mathcal T'}$ (\cite{schaefer_topological_1999}, Theorem 3.1, p. 130) and because of the above claim, $\overline{\text{co}_{\mathbb{Q}}(D)}^{\mathcal T'} = \overline{\text{co}(D)}^{\mathcal T'}$. Therefore, $\overline{\text{co}_{\mathbb{Q}}(D)}^{\mathcal T'} = Y$ so there exists a countable subset of $Y$ which is dense with respect to the topology $\mathcal T'$.\end{proof}

\ 

Let $\kappa\in\{\|\cdot\|, \tau, w\}$ be fixed. Let $F\subset L$ be a convex, $\kappa$-separable and metrizable set which is integrably bounded with respect to the finite measure space $(A, \mathcal A, \nu)$, in the sense that there exists an integrable function $h:A\rightarrow\mathbb R$ such that for every $g:A\rightarrow F$, it follows that $\|g(a)\|\le h(a)$ a.e. By Lemma \ref{separability} above, $F$ is $\|\cdot\|$-separable.

\subsection*{Bochner integral.}\ 

Throughout this section, we assume that $L' = L^{*}$ so Mackey and norm topologies coincide on $L$.

\begin{lemma}\label{smeasurable}
Let $\kappa\in\{\|\cdot\|, w\}$ and let $f:(A, \mathcal A)\rightarrow (F, \mathcal B(F^\kappa))$ be a measurable function. Then, $f$ is strongly measurable.
\end{lemma}
\begin{proof} For $\kappa = \|\cdot\|$, the result follows from Lemma \ref{measurestrongmeasure}. Let $\kappa  = w$ and notice, first, that $f(A)$ is a $\|\cdot\|$-separable subset of $F$ and $f(a)$ belongs to $F$ a.e. Second, we show that $f$ is weakly-measurable. Indeed, every $p\in L^*$ is $\left(\mathcal B\left(L^\kappa\right), \mathcal B\left(\mathbb R\right)\right)$-measurable (\cite{aliprantis_infinite_2006}, Corollary 4.26, p. 140) and $f$ is $\left(\mathcal A, \mathcal B\left(L^\kappa\right)\right)$-measurable since it is $\left(\mathcal A, \mathcal B\left(F^\kappa\right)\right)$-measurable and its range is in $F$. Hence, $p\cdot f$ is $\left(\mathcal A, \mathcal B\left(\mathbb R\right)\right)$-measurable  (\cite{hildenbrand_core_1974}, p. 42, (4)) so $f$ is weakly measurable, whence strongly by Theorem \ref{pettismeasurability}.\end{proof}

\begin{lemma}\label{Bintegration}
Let $\kappa\in\{\|\cdot\|, w\}$ and let $f:(A, \mathcal A)\rightarrow (F, \mathcal B(F^\kappa))$ be a measurable function. Then $f$ is Bochner integrable.   
\end{lemma}
\begin{proof} Let $\kappa = \|\cdot\|$. Since $F$ is $\|\cdot\|$-separable, there exists a sequence of $\left(\mathcal A, \mathcal B\left(F^{\|\cdot\|}\right)\right)$-measurable simple functions $(f_n)$ from $A$ into $F$ that converges to $f$ almost everywhere in the $\|\cdot\|$-topology (\cite{aliprantis_infinite_2006}, Theorem 4.38, p.~145).
By Remark~\ref{simplefunctionsbochner}, each $f_n$ is Bochner integrable. Moreover, $\|f_n(a)\| \le h(a)$ for some $h \in L^1(\nu)$ almost everywhere. Hence, by Theorem~\ref{DCT}, $f$ is Bochner integrable, since it is strongly measurable by Lemma~\ref{smeasurable}.

Now, let $\kappa = w$. As established in Lemma \ref{smeasurable}, the function $f$ is strongly measurable, from which we deduce the existence of a sequence of simple functions $f_n: A \to F$ such that 
$\text{lim}_n \|f_n(a) - f(a)\| = 0 \quad \text{for almost every } a \in A$
(see the proof of Theorem \ref{pettismeasurability} in \cite{diestel_vector_1977}, p.~42). The proof proceeds along the same lines as in the previous case.\end{proof}

\begin{proposition}\label{KDCT}

Let $\kappa\in\{\|\cdot\|, w\}$ and let $f$ and the sequence $(f_n)$ be measurable functions from $A$ into $F$ such that $\kappa-\textnormal{lim}_n f_n(a) = f(a)$ a.e., then $\kappa-\textnormal{lim}_n\int_A f_n (a)d\mu = \int_A f(a)d\mu$.
\end{proposition}

\begin{proof} First, let $\kappa = \|\cdot\|$. By Lemma \ref{Bintegration}, each $f_n$ is Bochner integrable and since $F$ is integrably bounded, one can use Theorem \ref{DCT} to get $$\|\cdot\|- \text{lim}_n \int_A f_n (a)d\mu =  \int_A f(a)d\mu$$ 

Second, take $\kappa = w$. We want to show that for every $p\in L'$, $$\text{lim}_n p\cdot\left (\int_A f_n (a)d\mu  
\right) = p\cdot\left (\int_A f (a)d\mu \right)$$ For any $p\in L'$, $\textnormal{lim}_n p\cdot f_n(a) = p\cdot f(a)$ a.e. By Lemma \ref{Bintegration}, both $f_n$ and $f$ are Bochner integrable and by Lemma \ref{lebesgueintegral}, both $p\cdot f_n$ and $p\cdot f$ are Lebesgue integrable. Since \( p \) is linear and continuous with respect to the weak topology $w$, and \( L \) is a Banach space, it follows that \( p \) is also $\|\cdot\|$-continuous (\cite{aliprantis_infinite_2006}, Theorem 6.17, p. 233), hence \( \|p\| < \infty \) (\cite{aliprantis_infinite_2006}, Lemma 6.4 (2), p. 229). Therefore, for each \( n \), the scalar functions \( p \cdot f_n \) and \( p \cdot f \) are measurable and satisfy
$|p\cdot f_n(a)| \le \|p\| \cdot \|f_n(a)\| \le \|p\| \cdot h(a) \quad \text{a.e.}$
where \( h \in L^1(\mu) \) is the integrable function dominating \( \|f_n\| \). Hence, by the Dominated Convergence Theorem for real-valued functions, we obtain $$\text{lim}_n \int_A p\cdot f_n(a) \, d\mu(a) = \int_A p\cdot f(a) \, d\mu(a)$$ Finally, by Lemma~\ref{lebesgueintegral}, we have
$$p\cdot \left( \int_A f_n(a) \, d\mu(a) \right) = \int_A p\cdot f_n(a) \, d\mu(a) \ \text{and} \ p\cdot \left( \int_A f(a) \, d\mu(a) \right) = \int_A p\cdot f(a) \, d\mu(a)$$
so it follows that
$$w\text{-}\text{lim}_n \int_A f_n(a) \, d\mu(a) = \int_A f(a) \, d\mu(a)$$

\end{proof}

From the previous lemmata, we derive the following corollaries, which will be used in the main results of the paper concerning Bochner economies. Recall that by Mackey's Theorem (see \cite{aliprantis_infinite_2006}, Theorem~6.20, p.~234, and Corollary~6.23, p.~236), a subset of \( F \) is \( w \)-bounded if and only if it is \( \|\cdot\| \)-bounded.

\begin{corollary}\label{cBintegration}
Let \( f \colon (A, \mathcal A) \to (F, \mathcal B(F^\kappa)) \) be a measurable function. If \( F \subset L \) is \( w \)-bounded, convex, and \( \kappa \)-Polish, then \( f \) is Bochner integrable.
\end{corollary}

\begin{corollary}\label{cKDCT}
Let \( f \) and the sequence \( (f_n) \) be measurable functions from \( A \) into \( F \), and suppose that \( f_n(a) \to f(a) \) in the \( \kappa \)-topology for almost every \( a \in A \). If \( F \subset L \) is \( w \)-bounded, convex, and \( \kappa \)-Polish, then
\[
\kappa\text{-}\text{lim}_n \int_A f_n(a)\, d\mu = \int_A f(a)\, d\mu.
\]
\end{corollary}

\subsection*{Gel'fand integral.}

\ 

We now consider that $F$ is $\|\cdot\|$-bounded and $L$ is a dual space of a Banach space $E$, i.e., $L = E^*$. By \cite{schaefer_topological_1999}, Theorem 1.6, p. 212 or \cite{aliprantis_infinite_2006}, Corollary 6.11 p. 231 for normed spaces, $E\subset E^{**}$ via the canonical embedding and thus $L^* = E^{**}$. Therefore, we identify $E$ with $L'\subset L^*$. Hence, $w = \sigma(L, L')$ is the weak$^*$ topology on $L$, which is weaker than $\tau = \tau(L, L')$, the Mackey topology on $L$. In turn, $\tau$ is weaker than the norm topology $\|\cdot\|$. Similarly, $w' = \sigma(L', L)$ is the weak topology on $L'$, which is weaker than the Mackey topology $\tau' = \tau(L', L)$. Moreover, in this case, $\tau'$ coincides with the norm topology $\|\cdot\|'$.

\begin{lemma}\label{Gintegration}
Let $\kappa\in\{\|\cdot\|, \tau, w\}$ and let $f:A\rightarrow F\subset L $ be a $(\mathcal A, \mathcal B(F^\kappa))$-measurable function. Then $f$ is Gel'fand integrable. 
\end{lemma}

\begin{proof} Notice that $f(A)$ is norm-bounded. In addition, we remark that for any $p\in L'$, the mapping $a\mapsto p\cdot f(a)$ is $(\mathcal A, \mathcal B(\mathbb R))$-measurable since $f$ is $(\mathcal A, \mathcal B (L^{\kappa}))$-measurable and $p$ is continuous. Therefore, $f$ is $L'$-measurable (or weak$^*$-measurable) and thus, Gel'fand integrable by Theorem \ref{Gel'fandintegral}.\end{proof}

\begin{proposition}\label{GDCT}
Let $ (f_n) $ be a sequence of $ (\mathcal{A}, \mathcal{B}(F^\kappa)) $-measurable functions from $ A $ into $ F $, and let $ f \colon A \to F $ also be $ (\mathcal{A}, \mathcal{B}(F^\kappa)) $-measurable. Suppose that $ f_n(a) \to f(a) $ in the $ \kappa $-topology a.e. Then $$\kappa-\textnormal{lim}_n\int_A f_n(a)d\mu = \int_A f(a)d\mu$$.
\end{proposition}

\begin{proof}

When $\kappa$ is the norm topology, the proof is a special case of Corollary \ref{cKDCT}. Let $\kappa = w$. By Lemma \ref{Gintegration}, $\int_A f_n(a)d\nu(a)$ and $\int_A f(a)d\nu(a)$ exist. We want to show that for every $p \in L'$, 
$\text{lim}_n \, p \cdot \left( \int_A f_n(a) \, d\nu(a) \right) = p \cdot \left( \int_A f(a) \, d\nu(a) \right)$ which is the definition of weak convergence in $L$. Because of Gel'fand integrability, $p\cdot \left(\int_A f_n (a)d\nu(a)\right) = \int_A p\cdot f_n (a)d\nu(a)$ for all $n\in\mathbb N$ and $p\cdot \left(\int_A f(a)d\nu(a)\right) = \int_A p\cdot f(a)d\nu(a)$. Since $F$ is $\|\cdot\|$-bounded, $p$ is continuous and for every $n\in\mathbb N$, $p\cdot f_n$ and $p\cdot f$ are measurable, we make use of the Lebesgue Dominated Convergence
Theorem\footnote{See Theorem 11.21 in \cite{aliprantis_infinite_2006}, p. 415. In this case, Theorem \ref{DCT} can also be applied, as both $p \cdot f_n$ and $p \cdot f$ are measurable functions with values in a norm-bounded subset of the Banach space $\mathbb{R}$.} to get $\textnormal{lim}_n \int_A p\cdot f_n (a)d\nu(a) = \int_A p\cdot f (a)d\nu(a)$. Hence, $w-\textnormal{lim}_n \int_A f_n (a)d\nu(a) = \int_A f(a)d\nu(a)$. 

Finally, consider $\kappa = \tau$. We want to prove that for any $w'$-compact, convex, circled subset of $L'$, $G$, $\text{lim}_n \ \underset{p\in G}{\text{sup}}\left\{\left|p\cdot\left(\int_A [f_n (a)- f(a)]d\nu(a)\right)\right|\right\} = 0$. Since $G$ is $w'$-bounded, it is $\|\cdot\|'$-bounded because $L = (L')^*$ (Mackey's Theorem in \cite{aliprantis_infinite_2006}, p. 234 for the dual pairing $\langle L', L\rangle$). Therefore, there exists $M > 0$ such that $\vert p\cdot f_n (a)\vert \le M$ for all $a\in A$, $p\in G$ and $n\in\mathbb N$. Furthermore, every $p\cdot f_n$ and $p\cdot f$ are measurable (given the continuity of each $p_n$ and $p$ and given the measurability of each $f_n$ and $f$) and $\text{lim}_n \ \underset{p\in G}{\text{sup}}\left\{\left|p\cdot\left(f_n (a)- f(a)\right)\right|\right\} = 0$ by hypothesis. By Theorem~\ref{UDCT}, we have:
\[
\text{lim}_n \ \underset{p\in G}{\text{sup}}\left\{\left| \int_A p\cdot\left[f_n (a)- f(a)\right]d\nu(a)\right|\right\} = 0
\] Because of Gel'fand integral, we conclude that $$\text{lim}_n \ \underset{p\in G}{\text{sup}}\left\{\left|p\cdot\left(\int_A [f_n (a)- f(a)]d\nu(a)\right)\right|\right\} = 0$$ \end{proof}

We finish this section by stating and proving a Gel'fand version of Corollary \ref{chvariable}.

\begin{lemma}\label{lem:gelfand_change_var}
Let $L$ be a dual Banach space, and let $f \colon (A, \mathcal A, \mu) \to (B, \mathcal B)$ be a measurable map. Let $g \colon B \to L$ be a function such that both $g$ and $g \circ f$ are Gel'fand integrable. Suppose that $\mathcal A = f^{-1}(\mathcal B)$ and that $\mu \circ f^{-1} = \nu$ for a measure $\nu$ on $\mathcal B$. Then, for every $B' \in \mathcal B$,
\[
\int_{B'} g \, d\nu = \int_{f^{-1}(B')} (g \circ f) \, d\mu.
\]
\end{lemma}

\begin{proof}
Let $p \in L'$. Since $g$ is Gel'fand integrable with respect to $\nu$, we have:
\[
p\cdot\left( \int_{B'} g \, d\nu \right) = \int_{B'} p\cdot g \, d\nu = \int_{B'} p(g) \, d\nu = \int_{B'} p\cdot g \, d(\mu \circ f^{-1})
\]

Now, applying the standard change-of-variables formula for the scalar-valued function $p\cdot g : B \to \mathbb{R}$ (see, e.g., \cite{hildenbrand_core_1974}, (36), p.~50), we obtain:
\[
\int_{B'} p\cdot g \, d(\mu \circ f^{-1}) = \int_{f^{-1}(B')} p\cdot g \circ f \, d\mu.
\]
Since $(p\cdot g \circ f)(a) = p\cdot g\cdot f(a)$, the right-hand side becomes:
\[
\int_{f^{-1}(B')} p\cdot(g \circ f) \, d\mu.
\]
Finally, because $g \circ f$ is Gel'fand integrable, 
\[
\int_{f^{-1}(B')} p\cdot(g \circ f) \, d\mu = p\cdot\left( \int_{f^{-1}(B')} (g \circ f) \, d\mu \right).
\]
Combining these equalities, we have shown that for every $p \in L'$,
\[
p\cdot\left( \int_{B'} g \, d\nu \right) = p\cdot\left( \int_{f^{-1}(B')} g \circ f \, d\mu \right).
\]
Since $L'$ separates points in $L$, the desired equality follows.
\end{proof}

As in the case of the change-of-variables formula for Bochner integrals, we have the following immediate corollary since $\mathcal A$ is a sub-$\sigma$-algebra of the Borel $\sigma$-algebra $\mathcal B(A)$.

\begin{corollary}\label{wchvariable}
Let $L$ be a dual Banach space, and let $f \colon (A, \mathcal B(A), \mu) \to (B, \mathcal B(B), \nu)$ be a measurable map. Let $g \colon B \to L$ be a function such that both $g$ and $g \circ f$ are Gel'fand integrable. If $\mu \circ f^{-1} = \nu$, then
\[
\int_{B} g \, d\nu = \int_{A} g \circ f \, d\mu.
\]
\end{corollary}
\end{appendices}

\end{document}